\documentclass[11pt]{article}

\usepackage[a4paper,margin=1in]{geometry}
\usepackage{authblk}
\usepackage{hyperref}
\usepackage{microtype}
\usepackage{siunitx}
\usepackage[T1]{fontenc}
\usepackage[utf8]{inputenc}
\usepackage{lmodern}
\usepackage{xspace}
\usepackage{float}
\usepackage[draft]{changes}

\usepackage[english]{babel}
\usepackage{algorithm}
\usepackage{algpseudocode}
\usepackage{booktabs}
\usepackage{multirow}
\usepackage{bm}
\usepackage{graphicx}
\usepackage{subcaption}

\usepackage{pgfplots}
\pgfplotsset{compat=1.18}

\usepackage{amsmath,amssymb,amsthm,amsfonts,mathtools}

\newcommand{\persistence}{{{\cal P}^{\star}_{\cal C}}}
\newcommand{\milanoEx}{\texttt{Milano}\xspace}

\newcommand{\ScaledNAP}{Scaled-NAP\xspace}
\newcommand{\sNAP}[1]{\ensuremath{\mathrm{sNAP}_{#1}}}

\DeclareRobustCommand{\SNAPEx}{%
	\ifmmode \sNAP{\alpha}\else \ScaledNAP\fi}

\newcommand{\SNAPcoll}{Stanford Network Analysis Project\xspace}

\definecolor{sbgreen}{RGB}{0, 153, 076}

\definechangesauthor[name={AV}, color=orange]{av}

\theoremstyle{definition}

\newtheorem{proposition}{Proposition}
\newtheorem{corollary}[proposition]{Corollary}
\newtheorem{remark}{Remark}

\title{Scaled Null-Adjusted Persistence:\\ A Multiscale Bridge between Modularity and Persistence}

\author[1]{Alessandro Avellone}
\author[2]{Paolo Bartesaghi}
\author[3]{Stefano Benati}
\author[1]{Rosanna Grassi}

\affil[1]{University of Milano - Bicocca, Via Bicocca degli Arcimboldi 8, 20126 Milano, Italy}

\affil[2]{University of Milano, Via Conservatorio 7, 20122 Milano, Italy}

\affil[3]{University of Trento, Via Verdi 26, 38122 Trento, Italy}

\begin{document}
	\maketitle
	
	\begin{abstract}
		Community detection methods must balance two competing objectives: identifying small, cohesive groups while avoiding excessive fragmentation. Modularity, the most widely adopted optimization criterion, typically merges small communities in large networks due to its resolution limit. In contrast, a persistence-based criterion promotes more granular partitions. We introduce Scaled Null-Adjusted Persistence (Scaled-NAP), a parametric family of quality functions that incorporates both these criteria. The definition exploits the exact identity between a cluster’s modularity contribution and its Null-Adjusted Persistence (NAP) multiplied by its relative volume. Raising this volume factor to a parameter $\alpha\in[0,1]$ yields NAP at $\alpha=0$ and modularity at $\alpha=1$, while intermediate values control the scale of the detected partition. We derive conditions under which merging two communities improves the objective function and characterize the emergence of scale dependence, including resolution-limit behaviour on Caveman graphs. We develop the Milano algorithm, a multilevel Louvain-style heuristic for optimizing Scaled-NAP on large networks. Experiments on weighted and unweighted Lancichinetti-Fortunato-Radicchi benchmarks show that Scaled-NAP achieves the highest or tied-highest recovery wherever the ground truth structure is detectable, with its advantage increasing under community-size heterogeneity. Tests on three real networks with up to 1.1 million nodes confirm its capability to identify fine-grained ground-truth communities. The Milano algorithm also turned out to be the fastest method evaluated on large networks. These results show that Scaled-NAP provides an effective and scalable bridge between modularity-based and persistence-based community detection methodologies.
	\end{abstract}
	
	
	\section{Introduction}
	\label{introduction}
	
	One of the main issues in network theory is to discover 
	the node
	community structure. It has important applications in many fields of operational research, such as pandemic contention \cite{Abdin2023, Baghersad2023}, warehouse management \cite{Hu2026}, terrorism investigation \cite{Lindelauf2013, Algaba2024}, finance \cite{Macmahon2015,Amini2024}, portfolio optimization \cite{zema2025}, asset allocation \cite{ditollo2026}, world trade \cite{grassi2021,bartesaghi2020}, voting alliances \cite{MacOn2012, Benati2026}, opinion surveys \cite{benati2024}. In a taxonomy proposed in \cite{Fortunato2016}, community detection methods can be constructive, such as variants of hierarchical clustering, they can rely on inferential statistics (assuming a data generative model), or simulate dynamic processes to reveal hidden structural information. Finally, one of the most important methodologies is to solve an optimization model, to find the communities as the outcome of an optimal partition problem. This method has a clear advantage over the others in its objectivity: given two partitions (corresponding to two community structures), one can establish which is the best by comparing their objective function. A crucial feature of optimization methods is identifying the most appropriate objective function to define a community in a given application.
	
	The modularity function is a widely used statistic (and objective function) to find the optimal partitions \cite{2006Newman}. The modularity compares the actual number of edges connecting node pairs within the same cluster with its expected value under a null hypothesis, called the configuration model. It is a widespread model, popularized through the Louvain algorithm available in the Python and R platforms, but is affected by the so-called resolution limit: small communities of a large network cannot be identified, as they are merged into larger groups owing to the mathematical properties of the modularity function \cite{2007Fortunato, Avellone2025}. To overcome this issue, different solutions have been proposed. Some literature suggested correcting the modularity through the arc density, see \cite{Costa2015}, or through the difference with a reference calculated through the $z$-score,  see \cite{Miyauchi2016}. Another possibility is to use different network statistics, such as the Map equation, see \cite{rosvall2008maps}, or the persistence function, see \cite{Piccardi2011}.

	In this contribution, we will focus our attention on the persistence function, see \cite{Piccardi2011, 2024Avellone, Avellone2025}, as a viable alternative to modularity to detect network communities. Indeed, as discussed in \cite{Avellone2025}, the persistence is defined by the probability (in a long term dynamic process) that a random walker moving randomly through the network will remain in the nodes of the same cluster. In \cite{Avellone2025}, a variation of the persistence, called Null-Adjusted Persistence (NAP), has been proposed. Taking inspiration from modularity, the persistence function is corrected by a term representing its expected value under the configuration model. The configuration model, that is, the same node degree distribution but arcs randomly located, is assumed as the null hypothesis: a random graph in which no community structure is imposed. So, the difference between the actual and the theoretical persistence is a measure that can reveal the community structure of the network. Indeed, the optimal partition, i.e. the community detection, is calculated through the maximum of the NAP. There are some important differences between NAP and modularity. From a theoretical point of view, it has been established that, conversely to modularity, the NAP does not suffer from the resolution limit. This property is important: affected by the network size, the modularity tends to produce large communities, to the point that, in empirical applications, the small clusters are not recognized, but merged into the largest communities. Conversely, NAP is insensitive to both network and community size, allowing large and small communities to remain clearly distinguishable regardless of the overall network scale. Empirical tests confirmed the theoretical properties in the practical applications. Modularity, as affected by the resolution limit, tends to merge the true small communities into the largest ones, while NAP, not affected, tends to recognize the quality of the small communities. Broadly speaking, modularity yields a smaller number of communities than NAP.

	However, having to choose between the two measures to detect communities, a researcher should  consider \textit{both} their quantitative properties \textit{and} the qualitative properties that are assumed by the two measures. Actually, NAP and modularity are two quantitative measures that rely on two different qualitative assumptions about what is a community. Modularity emphasizes a static interpretation, that is, nodes within a community are characterized by strong bonds (represented by arcs) between them. NAP emphasizes a dynamic interpretation, as communities are formed by nodes in which the information flow tends to be more persistent. One may wonder what could be done when the choice between the two measures is not straightforward, or when modularity is the appropriate measure but the resolution limit has to be avoided. In this work, a new parametric family of network statistics, called Scaled Null-Adjusted Persistence (\SNAPEx{}) is proposed as a bridge between NAP and modularity. Indeed, a parametric factor controls the resolution limit from the modularity to the persistence. The parametric factor corresponds to inflating the cluster volume, a measure defining the strength of the cluster. So, this parametric family depends on a real parameter $\alpha \in [0,1]$ that provides NAP and modularity as special cases: the NAP corresponds to $\alpha=0$, while the  modularity corresponds to $\alpha=1$. 
	
	In this contribution, we analyse the application of \SNAPEx{} to community detection from a theoretical and practical perspective. From a theoretical point of view, we will analyse the role of $\alpha$ to avoid the resolution limit. Having the two extreme measures modularity and NAP---the former affected by the resolution limit, the latter not---one may wonder what are the values of $\alpha$ that cause the resolution limit to disappear. We will see that small perturbation of $\alpha$ around the modularity are enough to obtain a resolution-free \SNAPEx{} measure in specific network structures while keeping
	very close to the actual modularity.  From a practical point of view, we elaborated on previous NAP optimization, see \cite{2024Avellone, Avellone2025}, to implement a heuristic procedure to find the optimal \SNAPEx{} partitions. The specific purpose is to design a scalable procedure capable of calculating near-optimal partitions without exceeding memory and time limits, even for graphs with billions of arcs. Therefore, taking advantage of past computational experiences, we developed a variant of the Louvain algorithm for the optimal \SNAPEx{} partition, an algorithm that is called \milanoEx, following the tradition of naming community subroutines with the name of a city, see \cite{blondel2008fast,traag2019louvain, bayan}. Next, we tested the reliability of the partitions calculated through \SNAPEx{} optimization on many ground-truth simulated networks and we found that \SNAPEx{} optimization can recover the hidden communities with remarkable precision, often improving the partitions calculated through alternative methods. We repeated the experiments on a set of real-world large networks with ground-truth, again showing the ability (and sometimes the weakness) of \SNAPEx{} optimization compared to other methods. 
	
	The paper is organized as follows. Section \ref{sec:nap} reviews Null-Adjusted Persistence and its relationship with modularity. Section \ref{sec:snap} introduces the \SNAPEx{} family and shows how NAP and modularity emerge as its two endpoint cases. Section \ref{sec:analytical} derives the conditions under which merging clusters improves the \SNAPEx{} objective and investigates scale dependence and resolution-limit behaviour, with specific reference to caveman graphs. Section \ref{sec:milano} addresses the resulting optimization problem and presents the multilevel \milanoEx algorithm. Section \ref{sec:benchmark} reports the computational experiments on weighted and unweighted synthetic benchmark networks, whereas Section \ref{sec:real-world} evaluates the proposed approach on large real-world networks with known community structure. Section \ref{sec:conclusions} summarizes the main findings and outlines directions for future research. Finally, Appendix \ref{app:ilp} presents an exact mixed-integer linear programming formulation of the problem, and Appendix \ref{app:measures} compares the external validation measures used in the computational analysis.
	
	\section{Null-adjusted persistence}
	\label{sec:nap}
	
	In this section, we recall the definition of the Null-Adjusted Persistence (NAP) introduced in \cite{Avellone2025}. 
	Let $G=(V, E, W)$ be a weighted, connected, and undirected graph where $V$ is the set of $n$ vertices (or nodes), $E$ is the set of $m$ edges, and $W$ is the set of edge weights. We denote by $\lvert V \rvert$ the cardinality of $V$. The subgraph induced by $U \subseteq V$ consists of the nodes in $U$ and all edges in $E$ having both endpoints in $U$ with their corresponding weights. 
	A clustering of $G$ is a vertex partition $\Pi=\{ \mathcal{C}_{1}, \dots, \mathcal{C}_{q}\}$ where $\mathcal{C}_{r} \subseteq V$, $\mathcal{C}_{r}\neq \emptyset$ and the induced subgraph $G_{\mathcal{C}_{r}}$ is connected for all $r=1,\ldots, q$. Partition elements $\mathcal{C}_{r}$ are referred to as clusters or communities.
	The set of all the admissible clusterings of a graph $G$ is denoted by $\wp(G)$. 
	\noindent Let ${\bf W}=[w_{ij}]$ be the weighted adjacency matrix, where $w_{ij}\in W$ denotes the weight of the edge connecting nodes $i$ and $j$ and $w_{ij}=0$ whenever no edge exists. We also denote by ${\bf A}=[a_{ij}]$ the associated adjacency matrix, defined by $a_{ij}=1$ if $w_{ij}>0$ and $a_{ij}=0$ otherwise. In the unweighted case, ${\bf W}={\bf A}$. The strength and degree of node $i$ are defined, respectively, as $s_i=\sum_{j=1}^{n} w_{ij}$ and $k_i=\sum_{j=1}^{n} a_{ij}$. We denote by $\mathbf{s}=(s_1,\ldots,s_n)$ and $\mathbf{k}=(k_1,\ldots,k_n)$ the strength and degree vectors, respectively.  Finally, $\operatorname{vol}({\mathcal C}) = \sum_{i\in {\mathcal C}} s_i$ is the total strength or the volume of a cluster ${\mathcal C}$, and $S=\sum_{i<j} w_{ij}$ denotes the total weight of the network. Accordingly, $\operatorname{vol}(V) = \sum_{i\in V} s_i=2S$ denotes the total volume of the network.
	
	The persistence probability (briefly, persistence) of a cluster ${\mathcal{C}}\subseteq V$ has been introduced by means of a lumped Markov chain in \cite{Piccardi2011}. Its definition can be rewritten in the following form: 
	
	\begin{equation} \label{prs1}
		{\cal P}_{\mathcal{C}}=\frac{\sum_{i,j\in {\mathcal{C}}}w_{ij}}{\sum_{i\in {\mathcal{C}}}s_{i}}
	\end{equation}
	
	Equation~\eqref{prs1} represents the fraction of the total strength of the nodes in ${\mathcal C}$ that is associated with edges whose endpoints both belong to ${\mathcal C}$. Equivalently, it can be interpreted as the probability that a random walker currently located in a node of ${\mathcal C}$ will remain in the same cluster after one step. 
	For an unweighted graph, the persistence probability becomes ${\cal P}_{\mathcal{C}}=\frac{\sum_{i,j\in {\mathcal{C}}}a_{ij}}{\sum_{i\in {\mathcal{C}}}k_{i}}$.
	
	The null-adjusted persistence $\persistence$ (NAP) of a cluster ${\mathcal{C}}\subseteq V$ (\cite{Avellone2025}) is defined as:
	
	\begin{equation}
		\label{definition1}
		\persistence=\frac{\sum_{i,j\in {\mathcal{C}}}w_{ij}}{\sum_{i\in {\mathcal{C}}}\sum_{j=1}^{n}w_{ij}}-\frac{\sum_{i,j\in {\mathcal{C}}}z_{ij}}{\sum_{i\in {\mathcal{C}}}\sum_{j=1}^{n}z_{ij}}
	\end{equation}
	
	where $\bf W$ is the weighted adjacency matrix of $G$ and $\bf Z$ is the weighted adjacency matrix of the null model, with $z_{ij}=\frac{s_{i}s_{j}}{2S}$.
	
	\noindent The first term in Eq. \eqref{definition1} is the actual persistence of ${\mathcal C}$ in the observed network data, whereas the second term represents the persistence of the same cluster of nodes under the configuration-model null hypothesis preserving the node-strength sequence, see \cite{Newman_Girvan2004,Newman2004}. Therefore, $\persistence$ measures the excess persistence, that is the extent to which the observed persistence of ${\mathcal C}$ exceeds what could be expected from node strength distribution alone: a positive value indicates that the cluster retains the random walker more strongly than would be expected in a randomized network and therefore exhibits a stronger-than-expected separation from the rest of the graph. 
	
	In an unweighted network, null-adjusted persistence $\persistence$ can be rewritten as
	
	\begin{equation}
		\persistence =
		\frac{2m_i}{2m_i+m_e}
		-
		\frac{2m_i+m_e}{2m},
		\label{definition2}
	\end{equation}
	
	where $m_i$ is the number of internal edges of cluster ${\mathcal C}$,
	$m_e$ is the number of edges connecting nodes in ${\mathcal C}$ to nodes outside
	${\mathcal C}$, and $m$ is the total number of edges in the network. Here, $2m_i+m_e$ corresponds to the volume $\operatorname{vol}({\mathcal C})$ of the cluster ${\mathcal C}$.
	
	Note that the same configuration model is also used to define modularity (see \cite{Arenas2008}). Indeed, the contribution of a cluster $\mathcal{C}$ to the modularity is defined as
	
	\begin{equation}
		\label{def_mod}    
		Q_{\mathcal{C}} = \frac{\sum_{i,j\in{\mathcal C}} w_{ij}}{2S}-\left(\frac{\sum_{i\in {\mathcal C}} s_i}{2S}\right)^2, 
	\end{equation}
	
	The modularity of a partition $\Pi=\{{\mathcal C}_1,\ldots,{\mathcal C}_q\}$ is then given by: 
	
	\begin{equation}
		\label{Mod}
		Q_{\Pi}
		=
		\sum_{{\mathcal C}\in\Pi}
		Q_{\mathcal C}.
	\end{equation}

	Similarly, for a partition $\Pi$, the total null-adjusted persistence is:
	
	\begin{equation}\label{NAP}
		{\cal P}^{\star}_{\Pi}
		=
		\sum_{{\mathcal C}\in\Pi}
		{\cal P}^{\star}_{\mathcal C}.
	\end{equation}
	
	Community detection can then be formulated as the maximization of the total NAP score over all admissible partitions:
	
	\begin{equation}
		\label{maxNAP}
		\max_{\Pi \in \wp(G)}
		{\cal P}^{\star}_{\Pi}
		=
		\max_{\Pi \in \wp(G)}
		\sum_{{\mathcal C}\in\Pi}
		{\cal P}^{\star}_{\mathcal C}.
	\end{equation}
	
	Computational results and comparison with other methods are available in \cite{Avellone2025}.
	
	\section{Scaled Null-Adjusted Persistence (Scaled-NAP)}
	\label{sec:snap}
	
	In this section, we introduce Scaled Null-Adjusted Persistence (\SNAPEx{}), a parametric family of objective functions that incorporates both null-adjusted persistence and modularity as special cases.
	
	Using the configuration-model expression $z_{ij}=\frac{s_i s_j}{2S}$  in equation \eqref{definition1}, we obtain:
	
	\begin{equation*}
		\sum_{i,j\in{\mathcal C}} z_{ij}
		=
		\frac{1}{2S}
		\left(\sum_{i\in{\mathcal C}} s_i\right)^2
		=
		\frac{\operatorname{vol}({\mathcal C})^2}{\operatorname{vol}(V)},
	\end{equation*}
	
	\begin{equation*}
		\sum_{i\in{\mathcal C}}\sum_{j=1}^n z_{ij}
		=
		\sum_{i\in{\mathcal C}}s_i\sum_{j=1}^n \frac{s_j}{2S}
		=
		\sum_{i\in{\mathcal C}} s_i
		=
		\operatorname{vol}({\mathcal C}).
	\end{equation*}
	
	Therefore, the null-adjusted persistence of ${\mathcal C}$ can be written as
	\begin{equation}
		{\cal P}^{\star}_{\mathcal C}
		=
		\frac{\sum_{i,j\in{\mathcal C}} w_{ij}}
		{\operatorname{vol}({\mathcal C})}
		-
		\frac{\operatorname{vol}({\mathcal C})}{\operatorname{vol}(V)}.
	\end{equation}
	
	Then, we define the relative volume of the cluster ${\mathcal C}$ as:
	
	\begin{equation}
		\pi_{\mathcal C}
		=
		\frac{\operatorname{vol}({\mathcal C})}{\operatorname{vol}(V)}
		=
		\frac{\sum_{i\in {\mathcal C}} s_i}{2S}.
	\end{equation}
	Multiplying the previous expression by the relative cluster volume yields:

	\begin{equation}
		\pi_{\mathcal C}{\cal P}^{\star}_{\mathcal C}
		=
		\frac{\sum_{i,j\in{\mathcal C}} w_{ij}}{\operatorname{vol}(V)}
		-
		\left(
		\frac{\operatorname{vol}({\mathcal C})}{\operatorname{vol}(V)}
		\right)^2.
	\end{equation}
	
	Recalling the expression of $Q_{\mathcal{C}}$ in \eqref{def_mod}, the right-hand side is exactly the contribution of cluster ${\mathcal C}$ to the weighted modularity.
	Hence,
	\begin{equation}
		\label{nap_modularity_relation}
		Q_{\mathcal C}
		=
		\pi_{\mathcal C}{\cal P}^{\star}_{\mathcal C}.
	\end{equation}
	
	Equation \eqref{nap_modularity_relation} reveals that NAP and modularity are two related objective functions. Their structural difference is how the relative cluster volume is weighted: modularity results from rescaling NAP by the relative volume of the cluster, while NAP is modularity divided by the relative volume of the cluster:
	\begin{equation}
		{\cal P}^{\star}_{\mathcal C}
		=
		\frac{Q_{\mathcal C}}{\pi_{\mathcal C}}.
	\end{equation}
	
	In other words, NAP can be interpreted as a density-normalized modularity contribution. This explains why, when used as an objective function, NAP tends to favour small, locally cohesive communities, whereas modularity, incorporating the relative cluster-volume factor in its equation, typically yields coarser partitions.
	
	This naturally raises the question of what happens when the influence of the relative cluster volume on the objective function can be explicitly controlled. We address this question by introducing a family of objective functions obtained by rescaling NAP by a power of the relative cluster volume.
	For $\alpha\in[0,1]$, we define the scaled null-adjusted persistence of cluster
	${\mathcal C}$ as
	
	\begin{equation}
		\label{snap_definition}
		\sNAP{\alpha}(\mathcal C)
		:=
		\pi_{\mathcal C}^{\alpha}
		{\cal P}^{\star}_{\mathcal C}.
	\end{equation}
	
	Equivalently,
	\begin{equation}
		\label{snap_weighted}
		\sNAP{\alpha}(\mathcal C)
		=
		\left(
		\frac{\operatorname{vol}({\mathcal C})}{\operatorname{vol}(V)}
		\right)^{\alpha}
		\left[
		\frac{\sum_{i,j\in{\mathcal C}} w_{ij}}
		{\operatorname{vol}({\mathcal C})}
		-
		\frac{\operatorname{vol}({\mathcal C})}{\operatorname{vol}(V)}
		\right].
	\end{equation}
	
	The endpoint cases of this family are straightforward. When $\alpha=0$, no volume scaling is applied and \SNAPEx{} coincides with NAP:
	
	\begin{equation}
		\sNAP{0}(\mathcal C)
		=
		{\cal P}^{\star}_{\mathcal C}.
	\end{equation}
	
	When $\alpha=1$, the full relative-volume scaling is applied and \SNAPEx{} coincides with the modularity contribution:
	
	\begin{equation}
		\sNAP{1}(\mathcal C)
		=
		\pi_{\mathcal C}{\cal P}^{\star}_{\mathcal C}
		=
		Q_{\mathcal C}.
	\end{equation}
	
	For a fixed value of $\alpha$, the definition is naturally extended from clusters to partitions by introducing the \SNAPEx{} score of a partition
	$\Pi=\{{\mathcal C}_1,\ldots,{\mathcal C}_q\}$ as
	
	\begin{equation}
		\label{snap_partition}
		\sNAP{\alpha}(\Pi)
		=
		\sum_{{\mathcal C}\in\Pi}
		\sNAP{\alpha}(\mathcal C).
	\end{equation}
	
	In particular, 
	\begin{equation}
		\label{snap_partition2}
		\sNAP{0}(\Pi)
		=
		{\cal P}^{\star}_{\Pi}. 
	\end{equation}
	
	and 
	\begin{equation}
		\sNAP{1}(\Pi)
		=
		{Q}_{\Pi} .
	\end{equation}
	
	Thus, NAP and modularity are recovered as the two endpoint objective functions of the \SNAPEx{} family.
	
	The corresponding community detection problem consists of finding the partition that maximizes the total \SNAPEx{} score over the set of all feasible partitions:
	
	\begin{equation}
		\label{snap_optimization}
		\max_{\Pi \in \wp(G)}
		\sNAP{\alpha}(\Pi)
		=
		\max_{\Pi \in \wp(G)}
		\sum_{{\mathcal C}\in\Pi}
		\sNAP{\alpha}(\mathcal C).
	\end{equation}

	Equation \eqref{snap_partition} defines a family of clustering statistics, and, when they are used as the objective function in \eqref{snap_optimization}, varying $\alpha$ controls the influence of the relative cluster volume on the optimal partition. 
	Small values of $\alpha$ make the index close to NAP, so we could expect an optimal partition composed of small-sized communities, while, as	$\alpha$ increases, the index is closer to modularity and therefore larger communities are promoted. This flexibility can be useful in applications where the appropriate scale of analysis cannot be determined \textit{a priori}.
	
	Finally, we report equations for the case of unweighted networks. As $\operatorname{vol}({\mathcal C})=2m_i+m_e$ and $\operatorname{vol}(V)=2m$, the relative volume becomes: 
	
	\begin{equation}
		\pi_{\mathcal C}
		=
		\frac{\operatorname{vol}({\mathcal C})}{\operatorname{vol}(V)}
		=
		\frac{2m_i+m_e}{2m}.
	\end{equation}

	In this case, \SNAPEx{} reduces to
	\begin{equation}
		\label{snap_binary}
		\sNAP{\alpha}(\mathcal C)
		=
		\left(
		\frac{2m_i+m_e}{2m}
		\right)^{\alpha}
		\left(
		\frac{2m_i}{2m_i+m_e}
		-
		\frac{2m_i+m_e}{2m}
		\right),
		\qquad
		\alpha\in[0,1].
	\end{equation}
	
	Accordingly, the two endpoint cases become
	\begin{equation}
		\sNAP{0}(\mathcal C)
		=
		{\cal P}^{\star}_{\mathcal C},
		\qquad
		\sNAP{1}(\mathcal C)
		=
		\frac{m_i}{m}
		-
		\left(
		\frac{2m_i+m_e}{2m}
		\right)^2
		=
		Q_{\mathcal C}.
	\end{equation}

	\section{Analytical Properties of the Scaled-NAP family}
	\label{sec:analytical}
	
	In this section, we investigate some analytical properties of the \SNAPEx{} family. We first derive a general condition under which merging two clusters increases the \SNAPEx{} objective function. 
	We then use this condition to investigate how the parameter $\alpha$ affects community aggregation.
	
	We denote by
	\begin{equation*}
		w_{\mathrm{in}}({\mathcal C})
		=
		\frac{1}{2}
		\sum_{i,j\in{\mathcal C}} w_{ij}
	\end{equation*}
	
	the total internal weight of ${\mathcal C}$, and, for two disjoint clusters ${\mathcal C}_1$ and
	${\mathcal C}_2$ we denote by  
	\begin{equation*}
		w_{12}
		=
		\sum_{i\in{\mathcal C}_1}
		\sum_{j\in{\mathcal C}_2}
		w_{ij}
	\end{equation*}
	
	the total weight of the edges connecting ${\mathcal C}_1$ and ${\mathcal C}_2$.
	
	With this notation, the \SNAPEx{} contribution of a cluster ${\mathcal C}$ can be written as
	\begin{equation}
		\sNAP{\alpha}(\mathcal C)
		=
		\left(
		\frac{\operatorname{vol}({\mathcal C})}{\operatorname{vol}({V})}
		\right)^{\alpha}
		\left[
		\frac{2w_{\mathrm{in}}({\mathcal C})}{\operatorname{vol}({\mathcal C})}
		-
		\frac{\operatorname{vol}({\mathcal C})}{\operatorname{vol}({V})}
		\right],
		\qquad \alpha\in[0,1].
	\end{equation}
	or, equivalently,
	\begin{equation}
		\label{snap_weighted_compact}
		\sNAP{\alpha}(\mathcal C)
		=
		\frac{2w_{\mathrm{in}}({\mathcal C})\operatorname{vol}({\mathcal C})^{\alpha-1}}
		{(\operatorname{vol}({V}))^{\alpha}}
		-				\left(\frac{\operatorname{vol}({\mathcal C})}
		{\operatorname{vol}({V})}\right)^{\alpha+1}.
	\end{equation}
	For unweighted networks, the quantities appearing in \eqref{snap_weighted_compact} admit a direct edge-count representation:
	\begin{equation}
		w_{\mathrm{in}}({\mathcal C})=m_i,
		\qquad
		\operatorname{vol}({\mathcal C})=2m_i+m_e,
		\qquad
		\operatorname{vol}({V})=2m.
	\end{equation}
	Similarly, for two disjoint clusters ${\mathcal C}_1$ and ${\mathcal C}_2$, the
	quantity $w_{12}$ coincides with the number of edges connecting the two clusters, usually denoted by $m_e^{(12)}$ (see \cite{Avellone2025}).
	
	\subsection{Merging clusters}
	
	In \cite{Avellone2025} we discussed the following properties: the convenience of merging two clusters (and therefore finding optimal communities of different sizes) may depend on $S$, the total weight of the network, as is the case of modularity, or be independent of it, as is the case of NAP. In the former case, we say that modularity is scale-dependent, since increasing the network size makes merging two communities more convenient. In the latter case, we say that NAP is scale-invariant, because the convenience of merging two clusters is independent of $S$. These properties provide a structural explanation for why optimal NAP partitions can detect small clusters that are instead merged into larger communities when optimal partitions are obtained by modularity. Therefore, a critical issue in \eqref{snap_optimization} is determining whether merging two clusters improves the value of the objective function and how this condition depends on $\alpha$. We first consider weighted networks. Let ${\mathcal C}_1$ and ${\mathcal C}_2$ be two disjoint clusters such that ${\mathcal C}_1\cup {\mathcal C}_2$
	induces a connected subgraph, and denote their internal weights and volumes, respectively, by
	
	\begin{equation*}
		u_r=w_{\mathrm{in}}({\mathcal C}_r),
		\qquad
		v_r=\operatorname{vol}({\mathcal C}_r),
		\qquad r=1,2.
	\end{equation*}
	
	\noindent If ${\mathcal C}={\mathcal C}_1\cup{\mathcal C}_2$ denotes the merged cluster, then
	
	\begin{equation*}
		w_{\mathrm{in}}({\mathcal C})=u_1+u_2+w_{12},
		\qquad
		\operatorname{vol}({\mathcal C})=v_1+v_2.
	\end{equation*}
	
	\noindent For $\alpha\in[0,1]$, we define the merging variation
	\begin{equation*}
		\Delta \sNAP{\alpha}
		=
		\sNAP{\alpha}(\mathcal C)
		-
		\sNAP{\alpha}({\mathcal C}_1)
		-
		\sNAP{\alpha}({\mathcal C}_2).
	\end{equation*}
	
	Merging ${\mathcal C}_1$ and ${\mathcal C}_2$ improves the objective function  whenever
	$\Delta \sNAP{\alpha}>0$.
	The following proposition provides an explicit threshold condition for the merge.
	
	\begin{proposition}
		\label{snap_merging_weighted}
		
		Let ${\mathcal C}_1$ and ${\mathcal C}_2$ be two disjoint clusters in a weighted undirected network such that ${\mathcal C}_1\cup{\mathcal C}_2$ induces a connected subgraph. For every $\alpha\in[0,1]$, merging ${\mathcal C}_1$ and ${\mathcal C}_2$ increases the \SNAPEx{} objective, i.e. $\Delta \sNAP{\alpha}>0 $, if and only if
		
		\begin{equation}
			\label{snap_weighted_threshold}
			w_{12} > T_{\alpha}({\mathcal C}_1,{\mathcal C}_2),
		\end{equation}
		
		where the merging threshold $T_{\alpha}({\mathcal C}_1,{\mathcal C}_2)$ is given by:
		\begin{equation}
			\label{snap_weighted_threshold_def}
			\begin{split}
				T_{\alpha}({\mathcal C}_1,{\mathcal C}_2)
				={}&
				u_1\left[
				\left(\frac{v_1+v_2}{v_1}\right)^{1-\alpha}-1\right]
				+ u_2
				\left[\left(\frac{v_1+v_2}{v_2}
				\right)^{1-\alpha}-1\right] \\
				&+\frac{(v_1+v_2)^{\alpha+1}
					-v_1^{\alpha+1}-
					v_2^{\alpha+1}}
				{4S\,(v_1+v_2)^{\alpha-1}
				}.
			\end{split}
		\end{equation}
	\end{proposition}
	
	\begin{proof}
		Recall from Equation~\eqref{snap_weighted_compact} that
		\begin{equation*}
			\sNAP{\alpha}(\mathcal C)
			=
			\frac{2w_{\mathrm{in}}({\mathcal C})
				\operatorname{vol}({\mathcal C})^{\alpha-1}}
			{(2S)^{\alpha}}
			-\frac{\operatorname{vol}({\mathcal C})^{\alpha+1}}
			{(2S)^{\alpha+1}}.
		\end{equation*}
		
		For the merged cluster ${\mathcal C}={\mathcal C}_1\cup{\mathcal C}_2$, we have
		\begin{equation*}
			\sNAP{\alpha}(\mathcal C)
			=
			\frac{
				2(u_1+u_2+w_{12})(v_1+v_2)^{\alpha-1}}{(2S)^{\alpha}}
			-	\frac{
				(v_1+v_2)^{\alpha+1}}	{(2S)^{\alpha+1}}.
		\end{equation*}
		
		Similarly, for $r=1,2$,	\begin{equation*}
			\sNAP{\alpha}({\mathcal C}_r)
			=
			\frac{	2u_r v_r^{\alpha-1}	}	{(2S)^{\alpha}}
			-\frac{	v_r^{\alpha+1}}{(2S)^{\alpha+1}}.
		\end{equation*}
		
		Therefore,
		\begin{equation*}
			\begin{split}
				\Delta \sNAP{\alpha}
				={}&
				\frac{2(u_1+u_2+w_{12})(v_1+v_2)^{\alpha-1}
					-	2u_1v_1^{\alpha-1}
					-	2u_2v_2^{\alpha-1}}
				{(2S)^{\alpha}} \\
				&-	\frac{
					(v_1+v_2)^{\alpha+1}
					-	v_1^{\alpha+1}
					- v_2^{\alpha+1}}
				{(2S)^{\alpha+1}}.
			\end{split}
		\end{equation*}
		Since $(2S)^{\alpha+1}>0$, the condition
		$\Delta \sNAP{\alpha}>0$ is equivalent to
		\begin{equation*}
			\begin{split}
				&2(2S)(u_1+u_2+w_{12})(v_1+v_2)^{\alpha-1}
				-
				2(2S)u_1v_1^{\alpha-1}
				-
				2(2S)u_2v_2^{\alpha-1} \\
				&\hspace{4cm}>
				(v_1+v_2)^{\alpha+1}
				-
				v_1^{\alpha+1}
				-
				v_2^{\alpha+1}.
			\end{split}
		\end{equation*}
		Isolating $w_{12}$ gives
		\begin{equation*}
			\begin{split}
				w_{12}
				>{}&
				u_1
				\frac{v_1^{\alpha-1}}{(v_1+v_2)^{\alpha-1}}
				+
				u_2
				\frac{v_2^{\alpha-1}}{(v_1+v_2)^{\alpha-1}}
				-
				(u_1+u_2) \\
				&+
				\frac{
					(v_1+v_2)^{\alpha+1}
					-
					v_1^{\alpha+1}
					-
					v_2^{\alpha+1}
				}
				{
					4S(v_1+v_2)^{\alpha-1}
				}.
			\end{split}
		\end{equation*}
		Equivalently, it can be rewritten as
		\begin{equation*}
			\begin{split}
				w_{12}
				>{}&
				u_1
				\left[
				\left(
				\frac{v_1+v_2}{v_1}
				\right)^{1-\alpha}
				-1
				\right]
				+
				u_2
				\left[
				\left(
				\frac{v_1+v_2}{v_2}
				\right)^{1-\alpha}
				-1
				\right] \\
				&+
				\frac{
					(v_1+v_2)^{\alpha+1}
					-
					v_1^{\alpha+1}
					-
					v_2^{\alpha+1}
				}
				{
					4S\,(v_1+v_2)^{\alpha-1}
				}.
			\end{split}
		\end{equation*}
		which is precisely condition~\eqref{snap_weighted_threshold}.
	\end{proof}
	
	We establish the corresponding result for unweighted networks:
	\begin{corollary}
		\label{snap_merging_binary}
		Let ${\mathcal C}_1$ and ${\mathcal C}_2$ be two disjoint clusters in an unweighted
		undirected network such that ${\mathcal C}_1\cup{\mathcal C}_2$ induces a connected
		subgraph. For every $\alpha\in[0,1]$, merging ${\mathcal C}_1$ and ${\mathcal C}_2$ increases the \SNAPEx{} objective, i.e. $\Delta \sNAP{\alpha}>0 $, if and only if
		\begin{equation}
			\label{snap_binary_threshold}
			m_e^{(12)}	>
			T_{\alpha}^{\mathrm{bin}}({\mathcal C}_1,{\mathcal C}_2),
		\end{equation}
		
		where $m_e^{(12)}$ denotes the number of edges connecting
		${\mathcal C}_1$ and ${\mathcal C}_2$, while, for $r=1,2$,
		$m_i^{(r)}$ and $m_e^{(r)}$ denote, respectively, the number of internal edges of ${\mathcal C}_r$ and the number of edges connecting ${\mathcal C}_r$ to the rest of the network.
		\begin{equation}
			\label{snap_binary_threshold_def}
			\begin{split}
				&T_{\alpha}^{\mathrm{bin}}({\mathcal C}_1,{\mathcal C}_2)\\
				&=	m_i^{(1)}
				\left[\left(\frac{2m_i^{(1)}+m_e^{(1)}+2m_i^{(2)}+m_e^{(2)}}{2m_i^{(1)}+m_e^{(1)}}
				\right)^{1-\alpha} -1\right] \\
				&+
				m_i^{(2)}
				\left[
				\left(
				\frac{
					2m_i^{(1)}+m_e^{(1)}+2m_i^{(2)}+m_e^{(2)}}{2m_i^{(2)}+m_e^{(2)}	}
				\right)^{1-\alpha}
				-1	\right] \\
				&+	\frac{\left(2m_i^{(1)}+m_e^{(1)}+2m_i^{(2)}+m_e^{(2)}\right)^{\alpha+1}
					-
					\left(2m_i^{(1)}+m_e^{(1)}\right)^{\alpha+1}
					-\left(2m_i^{(2)}+m_e^{(2)}\right)^{\alpha+1}}
				{4m \left(2m_i^{(1)}+m_e^{(1)}+2m_i^{(2)}+m_e^{(2)}\right)^{\alpha-1}	}.
			\end{split}
		\end{equation}
		
	\end{corollary}
	
	\begin{proof}
		The proof follows immediately from Proposition \ref{snap_merging_weighted} by observing that the total internal weight of a cluster coincides with the number of its internal edges, while its volume is the sum of the degrees of its nodes. Hence, for $r=1,2$, $u_r=m_i^{(r)}$ and $v_r=2m_i^{(r)}+m_e^{(r)}$. Moreover, the inter-cluster weight $w_{12}$ coincides with the number of edges connecting the two clusters, so that $w_{12}=m_e^{(12)}$ and the total weight of the network is the total number of edges, $S=m$.
		
	\end{proof}

	\begin{remark}
		The merging threshold in  \eqref{snap_weighted_threshold_def} correctly recovers the two endpoints of the family.
		\begin{itemize}
			\item If $\alpha=0$, then the last term 
			vanishes and merging does not depend on $S$:
			\begin{equation*}
				T_0({\mathcal C}_1,{\mathcal C}_2)
				=
				\frac{v_2}{v_1} u_1
				+
				\frac{v_1}{v_2} u_2
			\end{equation*}
			that is, in the unweighted case,
			\begin{equation*}
				m_e^{(12)}	> \frac{2m_i^{(2)}+m_e^{(2)}}{2m_i^{(1)}+m_e^{(1)}}\,m_i^{(1)}
				+
				\frac{2m_i^{(1)}+m_e^{(1)}}{2m_i^{(2)}+m_e^{(2)}}\,m_i^{(2)},
			\end{equation*}
			which is exactly formula (14) discussed in \cite{Avellone2025}.
			\item If $\alpha=1$, then $\sNAP{1}(\mathcal C)$ is the contribution of $\mathcal C$ to modularity, and the corresponding merging threshold depends on $S$:
			\begin{equation*}
				T_1({\mathcal C}_1,{\mathcal C}_2)=\frac{v_1v_2}{2S}
			\end{equation*}
			that is, in the unweighted case,
			\begin{equation*}
				m_e^{(12)}
				> \frac{\bigl(2m_i^{(1)}+m_e^{(1)}\bigr)\bigl(2m_i^{(2)}+m_e^{(2)}\bigr)}{2m},
			\end{equation*}
			which is precisely formula (15) discussed in \cite{Avellone2025}.		\end{itemize}
	\end{remark}
	
	Notice that the two endpoint cases exhibit different behaviours with respect to the total weight of the network $S$: the NAP merging threshold is independent of $S$, while the modularity threshold explicitly depends on it. 
	We now investigate how this scale dependence emerges within the \SNAPEx{} family. Indeed, Proposition~\ref{snap_merging_weighted} shows that the dependence on $S$ of merging two clusters appears as soon as $\alpha>0$. Since the third term of the sum in Eq. \eqref{snap_weighted_threshold_def} can be rewritten as
	\begin{equation*}
		\frac{
			(v_1+v_2)^{\alpha+1}
			-
			v_1^{\alpha+1}
			-
			v_2^{\alpha+1}
		}
		{
			4S\,(v_1+v_2)^{\alpha-1}
		}=\frac{(v_1+v_2)^2}{4S}\left[1-\left(\frac{v_1}{v_1+v_2}\right)^{\alpha+1}-\left(\frac{v_2}{v_1+v_2}\right)^{\alpha+1}\right]    
	\end{equation*}	
	
	For $\alpha>0$, this term is positive and inversely proportional to $S$. Hence, for fixed local properties of the two clusters, the merging threshold decreases as the total network weight $S$ increases.
	
	We now investigate the behaviour of the merging threshold in dependence on $\alpha$. We prove the following result:
	
	\begin{proposition}
		\label{snap_threshold_monotonicity}
		Let ${\mathcal C}_1$ and ${\mathcal C}_2$ satisfy the assumptions of
		Proposition~\ref{snap_merging_weighted}. If ${\cal P}^{\star}_{{\mathcal C}_1}>0$  and ${\cal P}^{\star}_{{\mathcal C}_2}>0$,
		then the merging threshold
		$T_{\alpha}({\mathcal C}_1,{\mathcal C}_2)$ is strictly decreasing
		in $\alpha\in[0,1]$.
	\end{proposition}   
	
	\begin{proof}
		
		We compute the first derivative of Eq. \eqref{snap_weighted_threshold_def} with respect to $\alpha$:
		\begin{equation}
			\begin{aligned}
				T'_{\alpha}(\mathcal C_1,\mathcal C_2)
				={}&
				u_1\left[-\left(\frac{v_1+v_2}{v_1}
				\right)^{1-\alpha}\ln\left(\frac{v_1+v_2}{v_1}
				\right)\right]
				+u_2\left[-\left(\frac{v_1+v_2}{v_2}
				\right)^{1-\alpha}\ln\left(\frac{v_1+v_2}{v_2}
				\right)\right]\\
				&+\frac{(v_1+v_2)^2}{4S}
				\left[-\left(\frac{v_1}{v_1+v_2}
				\right)^{\alpha+1}\ln\left(\frac{v_1}{v_1+v_2}\right)\right.
				\left.
				-\left(\frac{v_2}{v_1+v_2}\right)^{\alpha+1}
				\ln\left(\frac{v_2}{v_1+v_2}\right)\right].
			\end{aligned}
		\end{equation}
		
		This expression can be written as
		
		\begin{equation}
			\begin{aligned}
				T'_{\alpha}(\mathcal C_1,\mathcal C_2)
				={}&
				-u_1\left(\frac{v_1+v_2}{v_1}
				\right)^{1-\alpha}\ln\left(
				\frac{v_1+v_2}{v_1}\right)
				-u_2\left(\frac{v_1+v_2}{v_2}
				\right)^{1-\alpha}
				\ln\left(\frac{v_1+v_2}{v_2}
				\right)
				\\
				&+\frac{v_1+v_2}{4S}
				\left[v_1
				\left(\frac{v_1}{v_1+v_2}
				\right)^\alpha
				\ln\left(\frac{v_1+v_2}{v_1}
				\right)\right.\left.
				+v_2\left(\frac{v_2}{v_1+v_2}
				\right)^\alpha
				\ln\left(\frac{v_1+v_2}{v_2}
				\right)\right].
			\end{aligned}
		\end{equation}
		
		Collecting the terms associated with each cluster we obtain
		\begin{equation}
			\begin{aligned}
				T'_{\alpha}(\mathcal C_1,\mathcal C_2)
				=
				(v_1+v_2)
				\Bigg\{
				&
				\left(
				-\frac{u_1}{v_1}
				+\frac{v_1}{4S}
				\right)
				\ln\left(
				\frac{v_1+v_2}{v_1}
				\right)
				\left(
				\frac{v_1}{v_1+v_2}
				\right)^{\alpha}
				\\
				&+
				\left(
				-\frac{u_2}{v_2}
				+\frac{v_2}{4S}
				\right)
				\ln\left(
				\frac{v_1+v_2}{v_2}
				\right)
				\left(
				\frac{v_2}{v_1+v_2}
				\right)^{\alpha}
				\Bigg\}.
			\end{aligned}
		\end{equation}
		
		The sign of $T'_\alpha(\mathcal{C}_1,\mathcal{C}_2)$ therefore depends on the terms
		\begin{equation*}
			-\frac{u_r}{v_r}
			+\frac{v_r}{4S}=-\frac{1}{2}P_{\mathcal{C}_r}^{\star}
			,\qquad r=1,2.
		\end{equation*}
		
		We deduce that, if $P_{\mathcal{C}_1}^{\star}>0$ and $P_{\mathcal{C}_2}^{\star}>0$, then the threshold decreases as $\alpha$ increases. 
	\end{proof}
	
	As a numerical example, consider the unweighted network $G$ represented in Figure \ref{fig1}. 
	In this case, $m_i^{(1)}=4$, $m_e^{(1)}=7$, $m_i^{(2)}=3$, and $m_e^{(2)}=6$, so that $v_1=15$ and $v_2=12$. Therefore, the merging threshold corresponding to $\alpha=0$ is
	\begin{equation*}
		T_{0}^{\mathrm{bin}}({\mathcal C}_1,{\mathcal C}_2)
		=
		\frac{12}{15}\,4+\frac{15}{12}\,3
		=
		\frac{139}{20}
		=	6.95.
	\end{equation*}
	
	\begin{figure}[H]
		\centering
		
		\includegraphics[width=0.6\textwidth]{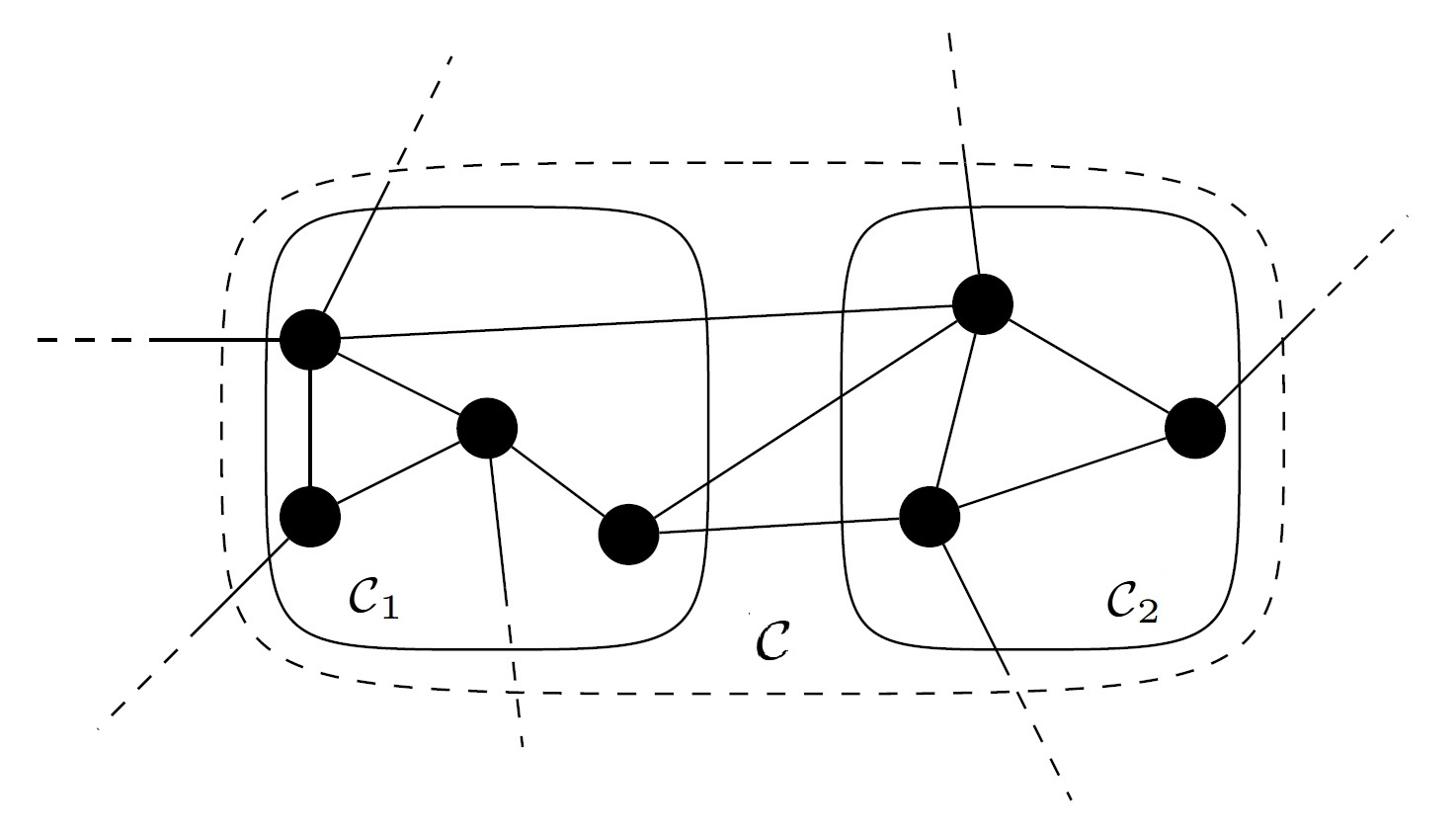}
		\caption{The merging condition of two clusters for the unweighted network $G$.}
		\label{fig1} 
	\end{figure}
	
	Hence, at least 
	$7$ inter-cluster edges are required   for merging to improve the NAP objective. Since the observed number of inter-cluster edges is
	$m_e^{(12)}=3$, the merging condition is not satisfied for $\alpha=0$.
	
	Figure~\ref{fig2} shows how the merging threshold varies with $\alpha$ for
	different values of the total number of edges $m$. For the three network
	sizes considered, both clusters have positive NAP:
	\begin{center}
		\begin{tabular}{c|cc}
			\hline
			$m$ & ${\cal P}^{\star}_{{\mathcal C}_1}$ 
			& ${\cal P}^{\star}_{{\mathcal C}_2}$ \\
			\hline
			20  & 0.158 & 0.200 \\
			50  & 0.383 & 0.380 \\
			100 & 0.458 & 0.440 \\
			\hline
		\end{tabular}
	\end{center}
	
	Thus, the assumptions of Proposition~\ref{snap_threshold_monotonicity} are satisfied for all three cases.
	Consistently with this result, the merging threshold decreases as $\alpha$ increases in all three cases, as shown in Figure \ref{fig2}.
	Moreover, for fixed $\alpha>0$, increasing the total number of edges $m$ further lowers the threshold. The figure illustrates the scale dependence of \SNAPEx{}: for the same local structure of the two clusters, different network sizes may lead to different merging decisions.
	
	\begin{figure}[ht]
		\centering
		\begin{tikzpicture}
			\begin{axis}[
				width=0.90\textwidth,
				height=0.60\textwidth,
				xmin=0, xmax=1,
				ymin=0, ymax=8,
				xlabel={$\alpha$},
				ylabel={$T_{\alpha}^{\mathrm{bin}}({\mathcal C}_1,{\mathcal C}_2)$},
				grid=both,
				grid style={line width=.1pt, draw=gray!25},
				major grid style={line width=.2pt, draw=gray!45},
				tick align=outside,
				tick style={black},
				axis line style={black},
				legend style={
					at={(0.25,0.35)},
					anchor=north east,
					draw=none,
					fill=white,
					fill opacity=0.85,
					text opacity=1
				},
				legend cell align={left},
				samples=200,
				domain=0:1
				]
				
				
				\addplot[very thick, blue]
				{
					4*(pow(27/15,1-x)-1)
					+
					3*(pow(27/12,1-x)-1)
					+
					(
					pow(27,x+1)-pow(15,x+1)-pow(12,x+1)
					)/(4*20*pow(27,x-1))
				};
				\addlegendentry{$m=20$}
				
				\addplot[very thick, teal!70!black]
				{
					4*(pow(27/15,1-x)-1)
					+
					3*(pow(27/12,1-x)-1)
					+
					(
					pow(27,x+1)-pow(15,x+1)-pow(12,x+1)
					)/(4*50*pow(27,x-1))
				};
				\addlegendentry{$m=50$}
				
				\addplot[very thick, orange!90!black]
				{
					4*(pow(27/15,1-x)-1)
					+
					3*(pow(27/12,1-x)-1)
					+
					(
					pow(27,x+1)-pow(15,x+1)-pow(12,x+1)
					)/(4*100*pow(27,x-1))
				};
				\addlegendentry{$m=100$}
				
				\addplot[densely dashed, very thick, red!70!black] coordinates {(0,3) (1,3)};
				\addlegendentry{$m_e^{(12)}=3$}
				
				\addplot[dotted, very thick, black] coordinates {(0,6.95) (1,6.95)};
				\addlegendentry{$T_0=139/20$}
				
				\node[anchor=west] at (axis cs:0.02,7.3)
				{\small NAP threshold};
				
				\node[anchor=west, red!70!black] at (axis cs:0.03,3.25)
				{\small observed inter-cluster edges};
				
			\end{axis}
		\end{tikzpicture}
		\caption{Merging threshold $T_{\alpha}^{\mathrm{bin}}({\mathcal C}_1,{\mathcal C}_2)$ as a function of $\alpha$ for the unweighted example with
			$m_i^{(1)}=4$, $m_i^{(2)}=3$, $m_e^{(1)}=7$, and $m_e^{(2)}=6$.
			For $\alpha=0$, the threshold is independent of the network size and equals
			$139/20=6.95$. For $\alpha>0$, the threshold depends on the total number of edges $m$ through a term proportional to $1/m$. The dashed horizontal line represents the observed value $m_e^{(12)}=3$.}
		\label{fig2}
	\end{figure}
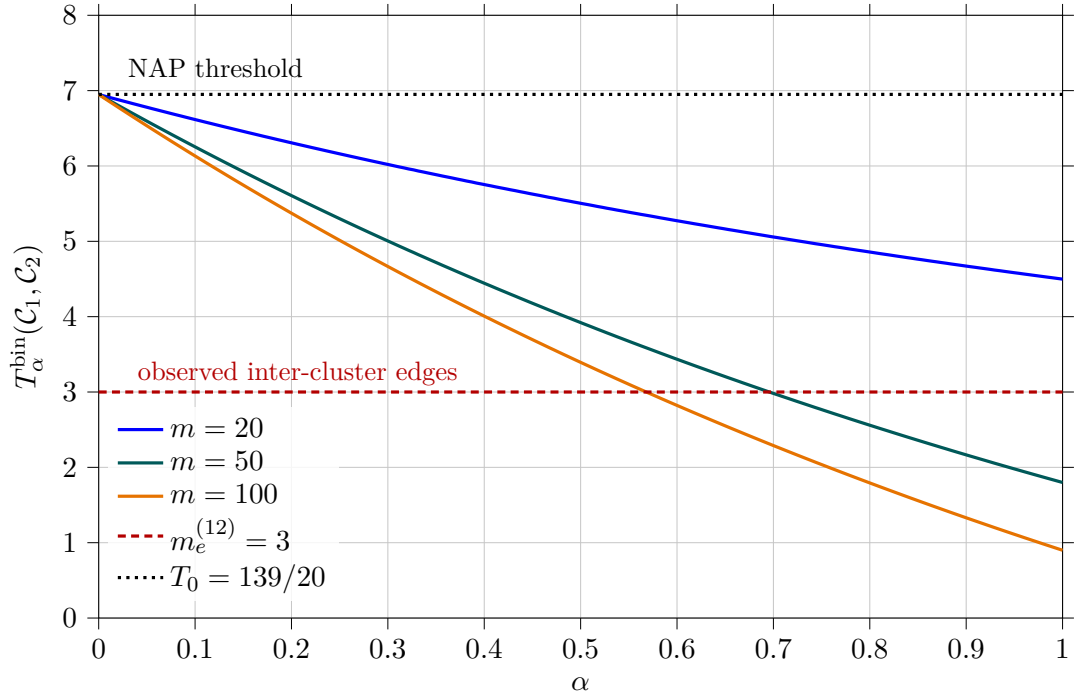
	
	\clearpage
	
	\subsection{Caveman graphs}
	
	A well-known drawback of modularity is the so-called resolution limit, namely the possibility of merging two communities that should instead be detected as separate, see \cite{2007Fortunato,Avellone2025}. 
	A classical network structure for illustrating this limit is provided by connected Caveman graphs, introduced in \cite{Watts1999} (see Figure \ref{fig:caveman}). These graphs are constructed from a collection of cliques representing tightly-knit communities, called \textit{caves}.
	Starting from each clique, one internal edge is removed and rewired to connect it to a neighbouring cave. By arranging these inter-cave edges cyclically, each cave is connected to exactly two neighbouring caves, with a single edge connecting each pair of adjacent caves. Modularity may favour partitions in which adjacent caves are merged, rather than preserving the individual caves as separate communities (see \cite{2007Fortunato}). Conversely, NAP preserves the natural partition into individual caves \cite{Avellone2025}.

	\begin{figure}[H]
		\centering
		\begin{tikzpicture}[
			scale=0.55,
			every node/.style={draw, circle, minimum size=0.36cm, inner sep=0pt}
			]
			\node (C11) at (0.00, 4.80) {1};
			\node (C12) at (0.80, 4.00) {2};
			\node (C13) at (0.00, 3.20) {3};
			\node (C14) at (-0.80, 4.00) {4};
			
			\node (C21) at (4.57, 1.48) {6};
			\node (C22) at (4.05, 0.48) {7};
			\node (C23) at (3.04, 0.99) {8};
			\node (C24) at (3.56, 2.00) {5};
			
			\node (C31) at (2.82, -3.88) {10};
			\node (C32) at (1.70, -3.71) {11};
			\node (C33) at (1.88, -2.59) {12};
			\node (C34) at (3.00, -2.77) {9};
			
			\node (C41) at (-2.82, -3.88) {14};
			\node (C42) at (-3.00, -2.77) {15};
			\node (C43) at (-1.88, -2.59) {16};
			\node (C44) at (-1.70, -3.71) {13};
			
			\node (C51) at (-4.57, 1.48) {18};
			\node (C52) at (-3.56, 2.00) {19};
			\node (C53) at (-3.04, 0.99) {20};
			\node (C54) at (-4.05, 0.48) {17};
			
			\draw (C11) -- (C12); \draw (C11) -- (C13); \draw (C11) -- (C14);
			\draw (C12) -- (C13); \draw (C12) -- (C14); \draw (C13) -- (C14);
			
			\draw (C21) -- (C22); \draw (C21) -- (C23); \draw (C21) -- (C24);
			\draw (C22) -- (C23); \draw (C22) -- (C24); \draw (C23) -- (C24);
			
			\draw (C31) -- (C32); \draw (C31) -- (C33); \draw (C31) -- (C34);
			\draw (C32) -- (C33); \draw (C32) -- (C34); \draw (C33) -- (C34);
			
			\draw (C41) -- (C42); \draw (C41) -- (C43); \draw (C41) -- (C44);
			\draw (C42) -- (C43); \draw (C42) -- (C44); \draw (C43) -- (C44);
			
			\draw (C51) -- (C52); \draw (C51) -- (C53); \draw (C51) -- (C54);
			\draw (C52) -- (C53); \draw (C52) -- (C54); \draw (C53) -- (C54);
			
			\draw[thick, dotted] (C12) -- (C24);
			\draw[thick, dotted] (C22) -- (C34);
			\draw[thick, dotted] (C32) -- (C44);
			\draw[thick, dotted] (C42) -- (C54);
			\draw[thick, dotted] (C52) -- (C14);
		\end{tikzpicture}
		\caption{Caveman graph with five caves of four nodes each.}
		\label{fig:caveman}
	\end{figure}
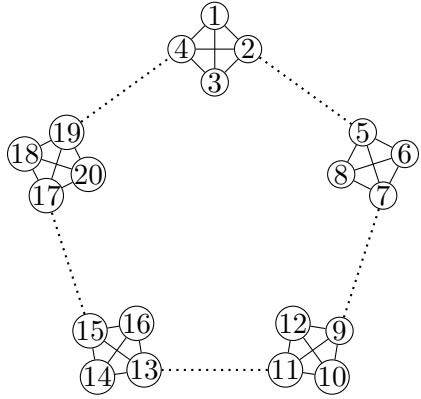
	
	Specifically,  \cite{Avellone2025} proves that the different behavior of NAP and modularity on Caveman graphs can be explained by their different dependence on network size. In particular, modularity is scale-dependent, whereas NAP is scale-invariant.
	
	Since equation \eqref{snap_definition} defines a family generalizing both NAP and modularity, it is meaningful to determine under which conditions on $\alpha$ and the network size the resolution limit may arise.
	
	
		
		
		
		
		
		
	
	\begin{proposition}
		\label{caveman_snap}
		Let $G=(V,E)$ be a connected Caveman graph with $q$ caves, each obtained from a clique of $p$ vertices by rewiring one internal edge to connect it to a neighboring cave. Assume $p\geq 3$. As $q\to+\infty$,  the merging threshold $T_{\alpha}^{\mathrm{bin}}$ converges to $\bigl(p(p-1)-2\bigr)
		\left(2^{1-\alpha}-1\right)$. In this asymptotic regime two adjacent caves are merged by the \SNAPEx{} criterion if and only if
		\begin{equation}
			\label{caveman_alpha_threshold}
			\alpha>\alpha^*=
			1-\log_2\!\left(
			\frac{p(p-1)-1}{p(p-1)-2}
			\right).
		\end{equation}
	\end{proposition}
	\begin{proof}
		Let ${\mathcal C}_1$ and ${\mathcal C}_2$ be two adjacent caves. The number of internal edges in each cave is
		\begin{equation*}
			m_i^{(1)}=m_i^{(2)}
			=
			\binom{p}{2}-1
			=
			\frac{p(p-1)}{2}-1.
		\end{equation*}
		Moreover, each cave has two external edges, so that
		\begin{equation*}
			m_e^{(1)}=m_e^{(2)}=2.
		\end{equation*}
		Therefore, the volume of each cave is
		\begin{equation*}
			v_1=v_2
			=
			2m_i^{(1)}+m_e^{(1)}
			=
			p(p-1).
		\end{equation*}
		Since the graph contains $q$ caves and the rewiring operation preserves the total
		number of edges of the original disjoint union of cliques, we have
		\begin{equation*}
			m
			=
			q\binom{p}{2}
			=
			\frac{q\,p(p-1)}{2}.
		\end{equation*}
		Finally, two adjacent caves are connected by a single edge, hence $m_e^{(12)}=1$. Using the  
		threshold in Equation \ref{snap_binary_threshold}, since $v_1=v_2=p(p-1)$, we obtain
		\begin{equation*}
			T_{\alpha}^{\mathrm{bin}}({\mathcal C}_1,{\mathcal C}_2)
			=
			2\left[\binom{p}{2}-1\right]
			\left(2^{1-\alpha}-1\right) 
			+
			\frac{
				(2p(p-1))^{\alpha+1}
				-
				2(p(p-1))^{\alpha+1}
			}
			{
				4m\,(2p(p-1))^{\alpha-1}
			}.
		\end{equation*}
		Since \(m=q\,p(p-1)/2\), the second term can be simplified as
		\begin{equation}
			\label{caveman_threshold_exact}
			T_\alpha^{\mathrm{bin}}({\mathcal C}_1,{\mathcal C}_2)
			=
			\bigl(p(p-1)-2\bigr)
			\left(2^{1-\alpha}-1\right)
			+
			\frac{2p(p-1)}{q}
			\left(1-2^{-\alpha}\right).
		\end{equation}
		
		By Corollary \ref{snap_merging_binary}, merging the two adjacent caves improves the objective function if and only if
		$m_e^{(12)} > T_\alpha^{\mathrm{bin}}({\mathcal C}_1,{\mathcal C}_2)$.
		Since $m_e^{(12)}=1$, in the asymptotic regime $q\to+\infty$ this condition becomes
		\begin{equation*}
			1>
			\bigl(p(p-1)-2\bigr)
			\left(2^{1-\alpha}-1\right).
		\end{equation*}
		Equivalently,
		\begin{equation*}
			2^{1-\alpha}
			<
			1+\frac{1}{p(p-1)-2}
			=
			\frac{p(p-1)-1}{p(p-1)-2}.
		\end{equation*}
		from which we finally obtain 			
		\begin{equation*}
			\alpha>
			1-\log_2\!\left(
			\frac{p(p-1)-1}{p(p-1)-2}
			\right)\coloneq \alpha^{\star}.
		\end{equation*}
	\end{proof}

	It is worth noting that $\alpha^\star<1$ for every $p\geq 3$,
	since the argument of the logarithm in
	\eqref{caveman_alpha_threshold} is strictly greater than $1$. Thus, in the asymptotic regime $q \to +\infty$, the resolution-limit behavior is not restricted to the modularity endpoint $\alpha=1$, but may arise for $\alpha<1$. 
	
	Moreover, the value $\alpha^\star < 1$ 
	depends on the cave size $p$ and approaches $1$ as $p$ increases. 
	Therefore, since
	\begin{equation*}
		\alpha^\star=
		1-\frac{1}{\ln2}\ln\!\left(
		\frac{p(p-1)-1}{p(p-1)-2}
		\right)
		=
		1-\frac{1}{\ln2}\ln\!\left(
		1+\frac{1}{p(p-1)-2}
		\right).
	\end{equation*}
	
	we can conclude that
	
	\begin{equation*}
		1-\alpha^\star
		\sim \frac{1}{p^2\ln 2}, \qquad p \to +\infty.
	\end{equation*}
	From an operational perspective, this result shows that, when $\alpha^\star$
	is close to $1$, \SNAPEx{} provides values of $\alpha<\alpha^\star$
	that remain close to the modularity endpoint while preserving the individual caves in this class of graphs.
	
	\begin{remark}
		Consider a Caveman network with $24$ cliques of $5$ nodes each. In this case,
		\[
		m_i^{(1)}=m_i^{(2)}=\binom{5}{2}-1=9,
		\qquad
		m_e^{(1)}=m_e^{(2)}=2,
		\qquad
		m=24\binom{5}{2}=240.
		\]
		Hence, $v_1 = v_2 := 2m_i^{(1)}+m_e^{(1)}=20$, and the \SNAPEx{} merging threshold is
		\begin{equation*}
			T_\alpha({\mathcal C}_1,{\mathcal C}_2)
			=
			9\left[2^{1-\alpha}-1\right]
			+
			9\left[2^{1-\alpha}-1\right]
			+
			\frac{
				40^{\alpha+1}
				-
				20^{\alpha+1}
				-
				20^{\alpha+1}
			}
			{
				4\cdot 240\cdot 40^{\alpha-1}
			}.
		\end{equation*}
		Equivalently,
		\begin{equation*}
			T_\alpha({\mathcal C}_1,{\mathcal C}_2)
			=
			18\left(2^{1-\alpha}-1\right)
			+
			\frac{5}{3}\left(1-2^{-\alpha}\right).
		\end{equation*}
		Since two adjacent caves are connected by a single edge, they are merged if
		$T_\alpha({\mathcal C}_1,{\mathcal C}_2)<1$:
		\begin{equation*}
			\alpha>\log_2\!\left(\frac{103}{52}\right)
			\approx 0.986061.
		\end{equation*}
	\end{remark}
	
	\section{The Milano Algorithm}
	\label{sec:milano}
	The optimization problem defined in \eqref{snap_optimization} is combinatorial: its feasible region is the set $\wp(G)$ of all partitions of $V$ into connected clusters. 
	Moreover, the \SNAPEx{} family 
	contains modularity as the special case $\alpha=1$, so maximizing \SNAPEx{} at $\alpha=1$ is exactly modularity maximization, which is NP-hard~\cite{brandes2008}. The problem therefore admits an NP-hard special case, and no exact polynomial-time algorithm can be expected for any $\alpha$.
	Regarding the null-adjusted persistence (NAP), that is the case $\alpha = 0$, 
	two solution methods were investigated in \cite{Avellone2025}. The first consists in finding the true optimal partitions through the solution of a Mixed-Integer Linear Programming (MILP) model, the second consists in finding near-optimal partitions through a heuristic procedure.  
	In \cite{Avellone2025}, the computational times of the former approach are prohibitive, while the near-optimal heuristic partitions are sufficiently accurate to retrieve  
	meaningful communities within short computational times.  
	Building on these computational results, in this work we extend the
	heuristic approach from NAP to the whole \SNAPEx{} family.
	Our heuristic for \SNAPEx{} community detection is called the \milanoEx algorithm and it is a straightforward extension of the Louvain procedure~\cite{blondel2008fast}. For the sake of completeness, we report in Appendix~\ref{app:ilp} a MILP formulation for \SNAPEx.
	
	The main procedure is listed in Algorithm \ref{alg:milano}.
	Line~\ref{alg:milano:init1} stands for the initialization of the data structures needed for the loop from Lines~\ref{alg:milano:loopstart} to~\ref{alg:milano:loopend}. 
	Line~\ref{alg:milano:init2} consists of assuming an initial partition $\Pi$ composed of singletons with the two corresponding (trivial) data structures: the total internal edge weight $w_{\mathrm{in}}(\mathcal{C}_k)$ and the volume $\operatorname{vol}({\mathcal{C}_k})$
	for every community $\mathcal{C}_k \in \Pi$.
	The partition is tentatively improved by the procedure \textsc{MoveNode} (described by Algorithm \ref{alg:movennode}) in Line~\ref{alg:milano:MoveNode}. If successful, then the input graph $G$ can be contracted into a smaller graph.  
	This  
	new graph $G'=(V',E')$ is characterized by nodes $V'$ representing the communities received as output in Line~\ref{alg:milano:MoveNode}.  
	More precisely, each community $ \mathcal{C}_k \in \Pi$ is represented by a super-node $u_k \in V'$. As a consequence,  
	the data structure representing 
	super-nodes $u_k$ and edges $(u_k, u_q)$ are defined as follows:
	\begin{itemize}
		\item $w(u_k) = w_{\mathrm{in}}(\mathcal{C}_k)$, i.e.\ the self-loop of $u_k$ carries the internal weight of $\mathcal{C}_k$; 
		\item $w_{kq} = w(u_k, u_q) = \sum_{i \in \mathcal{C}_k, j \in \mathcal{C}_q} w_{ij}$, that is, the weight $w_{kq}$ of the edge between nodes $u_k,u_q \in V'$ is the sum of the edge weights between clusters $\mathcal{C}_k$ and $\mathcal{C}_q$ in $G$.  
	\end{itemize}
	After these operations, $G'$ is itself a graph to which we apply the subroutine \textsc{MoveNode} once more, trying to move or aggregate super-nodes $u_k$ rather than the original nodes $u$. The procedure is repeated over successive contractions until no further improvement is found by \textsc{MoveNode}.
	Finally, in Line~\ref{alg:milano:verify} community connectivity is checked, as it has been documented that Louvain-style local search may end with disconnected communities.
	Note that Algorithm \ref{alg:milano} requires three more parameters beyond the graph $G$: they are the maximum number of sweeps $L_{\max}$ of the \textsc{MoveNode} subroutine (a \emph{sweep} is one full pass over all nodes), the threshold $\tau$ that is necessary to avoid numerical instability due to floating-point arithmetic (see Line~\ref{alg:movennode:tau} of \textsc{MoveNode}), and, finally, the random number generator $\rho$. 
	
	The subroutine \textsc{MoveNode} tries to improve the objective function by iteratively reallocating a node $u$ to some cluster $\mathcal{C}^*$. The reallocation is organized in Lines~\ref{alg:movennode:whilestart}-\ref{alg:movennode:whileend}. Firstly, nodes are shuffled in Line~\ref{alg:movennode:shuffled}  yielding a random ordering $\sigma$ of $V$ (using $\rho$). Then, following this order, beginning with the first, 
	each node is tentatively moved from its current cluster $\mathcal{C}_{\mathrm{old}}$ to the community $\mathcal{C}^*$ that yields the largest gain. 
	In Line~\ref{alg:movennode:best}, the best advantageous move for $u$ is stored and then $u$ is finally reallocated in Lines~\ref{alg:movennode:reallocateds}-\ref{alg:movennode:reallocatede}. The process is repeated for all $u$, see Lines~\ref{alg:movennode:fornodestart}-\ref{alg:movennode:fornodeend}. After that, nodes are reshuffled and the process is repeated in a new sweep. When, in a sweep, no reallocation is advantageous, counter $\mathit{nb\_moves}$ remains 0 from its initialization in Line~\ref{alg:movennode:nbmoves:init}, as Line~\ref{alg:movennode:nbmoves:inc} is not reached. Therefore, the condition on the While-loop beginning in Line~\ref{alg:movennode:whilestart} is not fulfilled and the loop is exited. A final reconstruction phase then rebuilds the partition $\Pi$ from the community assignments and removes the empty communities.
	
	Note that, as \textsc{MoveNode} is randomised through the node shuffling $\sigma$, the whole procedure lends itself to a multi-start scheme: several independent runs with different seeds can be performed and the partition with the highest \SNAPEx{} score is returned.

	\begin{algorithm}[t]
		\caption{\textsc{MilanoCommunities}}
		\label{alg:milano}
		\begin{algorithmic}[1]
			\Require Graph $G$, threshold $\tau$,
			max sweeps $L_{\max}$, RNG $\rho$
			\Ensure  Final community partition $\Pi$
			
			\State $\SNAPEx.\textsc{InitGraph}(G)$\;\label{alg:milano:init1}
			\State $\Pi \gets$ singleton partition;\;\label{alg:milano:init2}
			$G_{\mathrm{curr}} \gets G$;
			\Repeat\label{alg:milano:loopstart}
			\State $(\Pi,\, \mathit{impr})
			\gets \textsc{MoveNode}(G_{\mathrm{curr}},\, \Pi,\,
			\tau,\, L_{\max},\, \rho)$\label{alg:milano:MoveNode}
			\If{\textbf{not} $\mathit{impr}$} \textbf{break} \EndIf
			\State build contracted graph $G'$: each $\mathcal{C} \in \Pi$
			becomes a super-node; intra-community edges become self-loops;
			inter-community edges are aggregated by weight\label{alg:milano:buildgraph}
			\State $G_{\mathrm{curr}} \gets G'$;\; 
			$\Pi \gets$ singleton partition over $V(G')$\label{alg:milano:newinput}
			\Until{no improvement}\label{alg:milano:loopend}
			\State expand the super-nodes of $\Pi$ back to the original vertices $V$
			\State Verify the connectivity of $\Pi$.\label{alg:milano:verify}

			\State \Return $\Pi$
		\end{algorithmic}
	\end{algorithm}
	
	\begin{algorithm}[t]
		\caption{\textsc{MoveNode}}
		\label{alg:movennode}
		\begin{algorithmic}[1]
			\Require Graph $G$, initial partition $\Pi$, threshold $\tau$, max sweeps $L_{\max}$,
			RNG $\rho$
			\Ensure  Updated partition $\Pi$, flag $\mathit{improvement}$
			
			\State $\SNAPEx.\textsc{StartLevel}(\Pi)$
			\Comment{\textit{--- Phase A: local optimization ---}}
			\State $\mathit{improvement} \gets \mathbf{false}$;\;
			$\mathit{nb\_moves} \gets 1$;\; $s \gets 0$
			\While{$\mathit{nb\_moves} > 0$ \textbf{ and } ($L_{\max} = 0$ \textbf{ or }
				$s < L_{\max}$)}\label{alg:movennode:whilestart}
			\State $s \gets s+1$;\; $\mathit{nb\_moves} \gets 0$;\;
			shuffle $V$ into order $\sigma$ using $\rho$\label{alg:movennode:nbmoves:init}\label{alg:movennode:shuffled}
			\ForAll{$u \in \sigma$}\label{alg:movennode:fornodestart}
			\State $\mathcal{C}_{\mathrm{old}} \gets$ community of $u$;\;
			$\mathcal{C}^* \gets \mathcal{C}_{\mathrm{old}}$;\; $\Delta^* \gets 0$
			\ForAll{community $\mathcal{C} \ne \mathcal{C}_{\mathrm{old}}$ adjacent to $u$}\label{alg:movennode:forcomstart}
			\State $\Delta \gets \Delta \SNAPEx(u,\; \mathcal{C}_{\mathrm{old}} \to \mathcal{C})$
			\If{$\Delta - \Delta^* > \tau$}\; $\mathcal{C}^* \gets \mathcal{C}$;\;\label{alg:movennode:tau}\label{alg:movennode:best}
			$\Delta^* \gets \Delta$ \EndIf
			\EndFor\label{alg:movennode:forcomend}
			\If{$\mathcal{C}^* \ne \mathcal{C}_{\mathrm{old}}$}\label{alg:movennode:reallocateds}
			\State move $u$ to $\mathcal{C}^*$;\; $\SNAPEx.\textsc{Update}(u, \mathcal{C}_{\mathrm{old}}, \mathcal{C}^*)$
			\State $\mathit{improvement} \gets \mathbf{true}$;\;
			$\mathit{nb\_moves} \gets \mathit{nb\_moves}+1$\label{alg:movennode:nbmoves:inc}
			\EndIf\label{alg:movennode:reallocatede}
			\EndFor\label{alg:movennode:fornodeend}
			\EndWhile\label{alg:movennode:whileend}
			\Comment{\textit{--- Phase B: partition reconstruction ---}}
			\If{$\mathit{improvement}$}
			\State rebuild $\Pi$  from the final community assignments and remove all the empty communities\;
			\EndIf
			\State \Return $(\Pi,\;
			\mathit{improvement})$
		\end{algorithmic}
	\end{algorithm}

	\section{Community detection on simulated and real-world networks}
	\label{sec:benchmark}
	
	\subsection{Baseline algorithms and experimental measures}
	\label{subsec:baselines}
	We evaluate the performance, scalability, and robustness of the proposed \milanoEx{} algorithm, implemented in C++ and released as the R and Python packages \texttt{scalednap},
	against other community detection methods with different methodological approaches, all available as Python packages: 
	\begin{itemize}
		\item \textbf{Modularity optimization:} We considered the \textit{Louvain} algorithm~\cite{blondel2008fast}, the \textit{Leiden} algorithm~\cite{traag2019louvain}, and the Constant Potts Model (CPM)~\cite{traag2011narrow}. The first two maximize modularity, whose null model is the configuration model (Leiden being a computational improvement of the Louvain heuristic), while CPM compares each community's internal density against a constant threshold. Modularity with the configuration model is affected by the resolution limit, whereas CPM is not. All three objectives expose a resolution parameter, which we sweep over its range: for Louvain and Leiden it is the parameter $\gamma$ of the modularity-with-resolution (\texttt{RBConfiguration}) objective of Reichardt and Bornholdt~\cite{Reichardt2006}, while for CPM it is an internal edge-density threshold on its own scale.
		\item \textbf{Information theory:} The algorithm \textit{Infomap}~\cite{rosvall2008maps} uses information theory to define the objective function, i.e. the map equation, used to evaluate clusterings. While modularity can find communities in which internal arcs are more dense than expected, the map equation identifies communities in which a random walker tends to remain for relatively long periods. 
		\item \textbf{Random-walk distances:} The \textit{Walktrap} algorithm~\cite{pons2005computing} computes structural distances between nodes by simulating random walkers over 
		the networks nodes. Communities are aggregated hierarchically on this random-walk distance with a Ward-like criterion — each step merges the two adjacent communities that least increase the total within-community variance of the squared distances — and the dendrogram is cut at the level of maximum modularity. In the \textsc{igraph} implementation the transition-probability matrix is computed analytically, so Walktrap is deterministic; unlike 
		other optimization-based methods, it does not detect communities by optimizing an explicit objective function.  Modularity is instead used only to select the final cut of the dendrogram.
		\item \textbf{Algebraic flow simulation:} The \textit{Markov Clustering Algorithm (MCL)}~\cite{dongen2000graph} identifies cluster structures by simulating stochastic flows via alternate matrix expansion and inflation operations. 
		MCL has been proposed for dense, clique-like topologies and real-world biological networks. 
		The algorithm controls the community size with the inflation parameter. As for the resolution, we tested various alternative values.
		\item \textbf{Local consensus:} The \textit{Label Propagation Algorithm (LPA)}~\cite{raghavan2007near} is a constructive method that finds communities through an iterative local majority-voting process. Given its near-linear time complexity, 
		LPA serves as an essential baseline for bounding computational efficiency and 
		assessing robustness under high topological noise.
		
	\end{itemize}
	
	These Python packages are used:
	Leiden and its CPM variant via \texttt{leidenalg}~\cite{traag2019louvain}, Louvain via \textit{NetworkX}~\cite{hagberg2008exploring}, Walktrap and LPA via the native \textit{python-igraph} routines~\cite{csardi2006igraph}, Infomap via the official \textit{Infomap} package~\cite{rosvall2008maps}, and MCL via the \texttt{markov\_clustering} library. Methods that expose a resolution parameter (\milanoEx{}, Louvain, Leiden, CPM, MCL) are swept over that parameter and compared at equal number of clusters; Infomap, Walktrap and LPA, which yield a single partition, are reported at the cluster count they naturally produce.
	
	We test all algorithms on two classes of networks. The first consists of \emph{synthetic} graphs, both weighted and unweighted, generated by the LFR benchmark of Lancichinetti, Fortunato and Radicchi~\cite{2008LFR}. The ground-truth community structure is established by construction and some control parameters, such as the mixing parameter, the degree distribution and the distribution of community sizes, make clustering more or less difficult. Since the ground-truth partition is known, the recovery achieved by each method can be quantified and compared across algorithms. The second class consists of some \emph{real-world} networks from the \SNAPcoll  collection~\cite{YangLeskovec2015}, selected from the instances in which the ground-truth communities are at least partially known.
	
	The quality of the partition returned by each method is assessed with two external measures, depending on the nature of each benchmark. On the LFR networks, whose planted partition is a complete, non-overlapping ground truth, we report the Adjusted Mutual Information (AMI)\footnote{Since the commonly used Normalized Mutual Information (NMI)  
		is not corrected for chance and therefore it is systematically inflated on fine-grained partitions, an effect that grows with the number of clusters~\cite{vinh2010information}, we relied on the AMI for our comparison. 
		The magnitude of the bias is reported in  Appendix~{}\ref{app:measures}.}
	\cite{vinh2010information}, computed with the \texttt{sklearn.metrics} module of \textit{scikit-learn}~\cite{scikit-learn}, while on \SNAPcoll networks, whose ground-truth communities are partial and overlapping,
	we report the average best-matching $F_1$ score of Yang and Leskovec~\cite{YangLeskovec2013}, implemented in our codes.

	All computational experiments were performed on a workstation equipped with an Intel Core i9-14900K processor (24 cores, up to 6.0~GHz), 128~GB of DDR5 RAM (5600~MHz), and a 2~TB NVMe PCIe~4.0 SSD (Samsung 990~Pro), running Ubuntu 24.04.2 LTS (Linux kernel 6.17). All methods are CPU-only \footnote{We state this explicitly because several recent community-detection implementations do offer GPU-accelerated back-ends (for instance through \textsc{cuGraph}), so that timings obtained with and without such acceleration would not be comparable: none of the runs reported here uses the GPUs installed on the machine.}. Every single run is isolated in a dedicated process; on the real-world networks (Section~\ref{sec:real-world}) a wall-clock limit of one hour ($3600$~s) per run is set, a run exceeding this limit is reported as a \emph{timeout} (TO) and a run exhausting the available memory as \emph{out-of-memory} (OOM). 
	
	\subsection{Stability of the resolution parameter}
	\label{sec:alpha-coherence}
	
	The density exponent $\alpha$ in \SNAPEx{} acts as a \emph{resolution knob}: its largest values reward high-strength communities and thus favour coarser partitions, while its smallest values yield 
	finer structures. So, on a graph without ground truth, one can sweep $\alpha$ to generate a family of partitions at 
	different levels of granularity and inspect which aggregation level is the most meaningful for the domain---much as one varies the resolution $\gamma$ in modularity-based methods. The natural question is whether this sweep is \emph{reliable}: does a high-recovery region 
	emerge along $\alpha\in[0,1]$, and is its location stable across graphs of different difficulty and density? 
	We investigate these questions below on LFR benchmarks where the planted partition is known, and recovery can be measured.
	
	\paragraph{Unweighted networks}
	
	We generate LFR networks~\cite{2008LFR} with $N=10\,000$ nodes, community sizes in $[50,1000]$, degree exponent $\tau_1=2.0$ and community-size exponent $\tau_2=1.5$, the last two being the values recommended in~\cite{2008LFR}. Two parameters are varied. The first is the \emph{mixing parameter} $\mu$, the fraction of each node's edges that point outside its own community: we let $\mu\in\{0.3,0.4,0.5,0.6,0.7\}$, and the larger $\mu$, the harder it is to recover the planted partition. The second is the \emph{average degree} $\bar{k}\in\{10,25,50\}$ (with maximum degree $10\,\bar{k}$, i.e.\ $100$, $250$ and $500$ respectively), which controls density: the higher the density, the easier is the recovery, because each node then has a greater absolute number of edges pointing inside its own community, even at fixed $\mu$. All the remaining parameters are fixed. For each of the $5\times3=15$ parameter configurations, $20$ independent connected graphs are instantiated. For every graph, the \milanoEx algorithm is run, sweeping $\alpha$ from $0$ to $1$ in steps of $0.05$ ($21$ values), using $10$ independent multistart restarts per value and retaining the partition with the highest objective function value. Then, the reported AMI is the average over the $20$ instances of each parameter configuration.
	
	For each parameter configuration $(\mu,\bar{k})$, we define $\alpha^*$ as the value of $\alpha$ that maximises the average AMI over the $20$ instances. Table~\ref{tab:alpha-star} reports $\alpha^*$ together with the corresponding AMI$(\alpha^*)$, and Figure~\ref{fig:alpha-coherence} shows the same data as a heatmap (panel (a)) and as AMI$(\alpha)$ curves, one panel per average degree (panels (b)-(d)).

	\begin{table}[t]
		\centering
		\begin{tabular}{c ccc}
			\toprule
			& \multicolumn{3}{c}{average degree $\bar{k}$}\\
			\cmidrule(lr){2-4}
			$\mu$ & $10$ & $25$ & $50$\\
			\midrule
			$0.3$ & $0.95$ \,(0.992) & $0.90$ \,(1.000) & $0.85$ \,(1.000)\\
			$0.4$ & $0.95$ \,(0.886) & $0.95$ \,(1.000) & $0.85$ \,(1.000)\\
			$0.5$ & $0.95$ \,(0.706) & $0.95$ \,(0.999) & $0.85$ \,(1.000)\\
			$0.6$ & $0.90$ \,(0.378) & $0.95$ \,(0.976) & $0.85$ \,(1.000)\\
			$0.7$ & $0.85$ \,(0.144) & $0.90$ \,(0.745) & $0.85$ \,(0.991)\\
			\bottomrule
		\end{tabular}
		\caption{Computational results on unweighted graphs: optimal $\alpha^*$ for each $(\mu,\bar{k})$ configuration, with the corresponding AMI in parentheses.}
		\label{tab:alpha-star}
	\end{table}
	
	\begin{figure}[t]
		\centering
		\begin{subfigure}[b]{0.48\textwidth}
			\includegraphics[width=\textwidth]{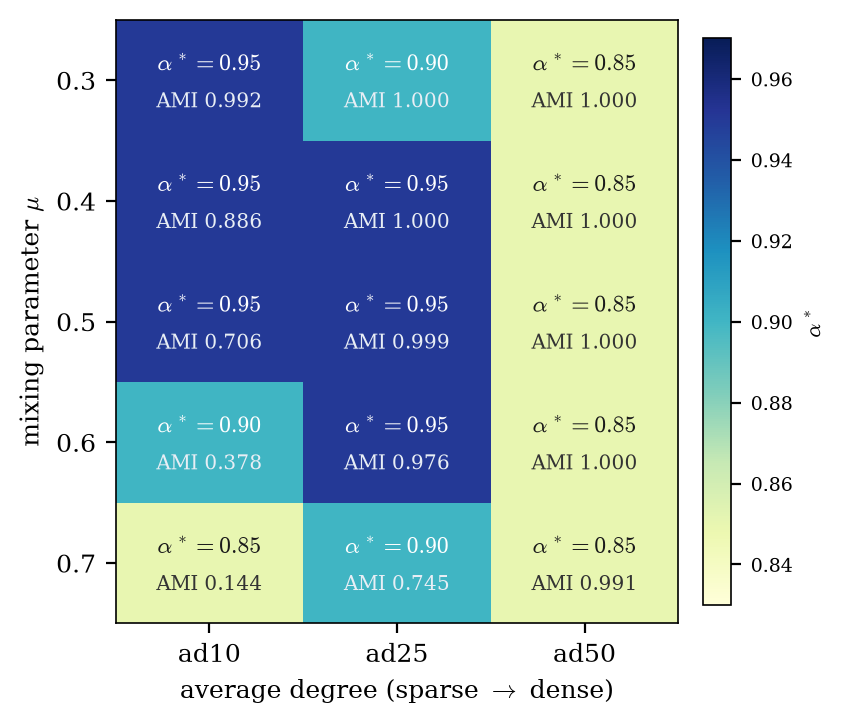}
			\caption{}
			\label{fig:alpha-heatmap}
		\end{subfigure}
		\hfill
		\begin{subfigure}[b]{0.48\textwidth}
			\includegraphics[width=\textwidth]{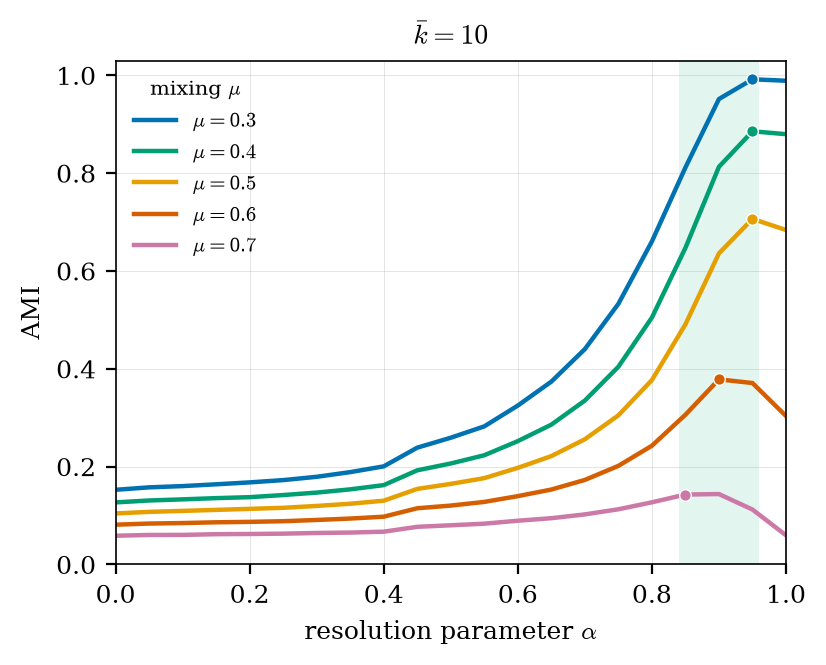}
			\caption{}
			\label{fig:alpha-curves-d10}
		\end{subfigure}
		
		\vspace{0.6em}
		\begin{subfigure}[b]{0.48\textwidth}
			\includegraphics[width=\textwidth]{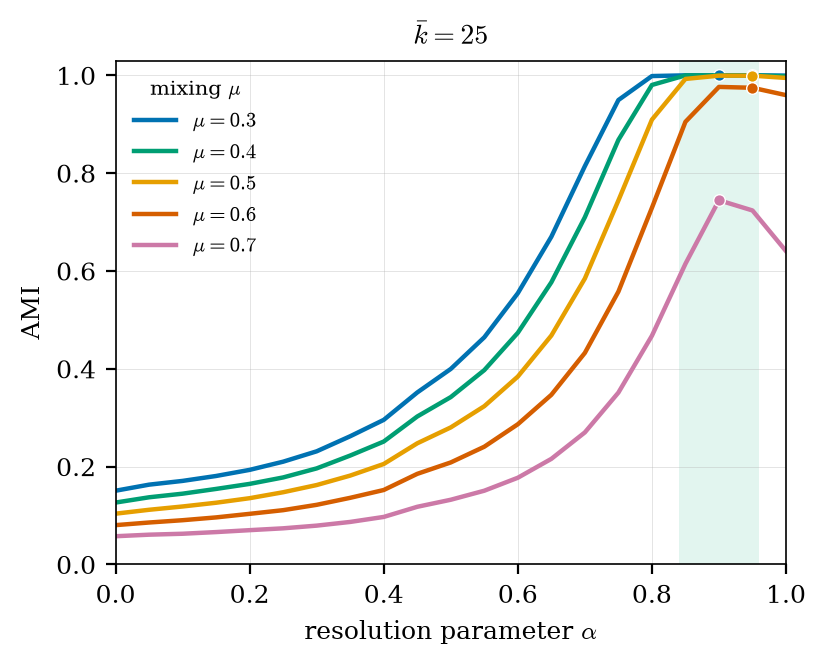}
			\caption{}
			\label{fig:alpha-curves-d25}
		\end{subfigure}
		\hfill
		\begin{subfigure}[b]{0.48\textwidth}
			\includegraphics[width=\textwidth]{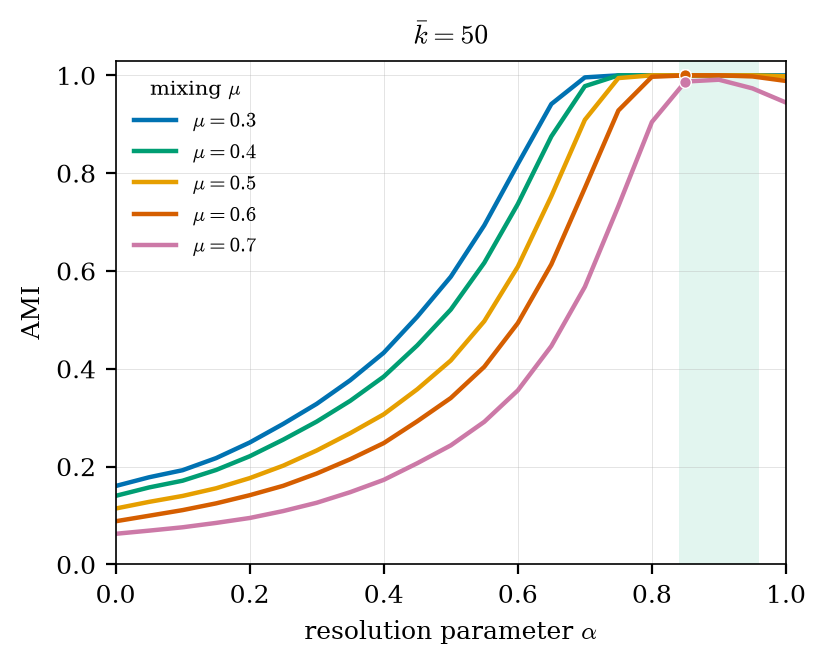}
			\caption{}
			\label{fig:alpha-curves-d50}
		\end{subfigure}
		\caption{Stability of the optimal $\alpha^*$ across the $(\mu,\bar{k})$ grid. \textbf{(a)}~$\alpha^*$ shown as a heatmap. \textbf{(b--d)}~AMI$(\alpha)$ curves at fixed average degree $\bar{k}=10,25,50$ respectively, the shaded band is the optimal $\alpha^*$ range $[0.85,0.95]$.}
		\label{fig:alpha-coherence}
	\end{figure}
	
	Three observations follow. First, as can be seen in Figure~\ref{fig:alpha-coherence}, panel (a), the value of $\alpha^*$ is remarkably \emph{stable}: across the fifteen configurations it lies in the interval $[0.85,0.95]$. 
	
	Second, the densest setting $\bar{k}=50$ pins $\alpha^*$ to the lower end of the band ($0.85$) in every row, without exception, whereas the sparser settings sit higher, within $[0.85,0.95]$. This suggests a possible interpretation: dense communities already have high internal strength, so only a small amount of volume scaling is enough to aggregate them. Conversely, sparser communities require slightly more volume scaling. 
	
	Third, as can be seen in panels (b)-(d), increasing $\mu$ lowers the \emph{height} of the maximum AMI but essentially not its \emph{location}. Actually, when the graph is dense, that is $\bar{k} = 50$, communities are recognized even when the mixing parameter $\mu$ is high, as can be seen from the AMI values close to 1. When the density is low, that is $\bar{k} = 10$, then the mixing parameter $\mu$ has an impact on the ability to recover the true communities, as the AMI decreases to $0.14$. However, as pointed out above,  the value $\alpha^*$ remains within the narrow interval $[0.85,0.95]$
	both in the cases in which communities are recognized and the cases in which they are not. 
	
	\paragraph{Weighted networks}
	We now investigate whether the stability observed for $\alpha^*$ carries
	over to weighted networks. We simulate LFR networks using the same basic
	parameter settings as in the unweighted case, but fixing  the average degree at $\bar{k} = 25$, maximum degree at $250$, with $N=10\,000$, $\tau_1=2.0$, $\tau_2=1.5$ and community sizes in $[50,1000]$. As before, the presence of edges between nodes of different communities is controlled by the
	parameter $\mu$, while their weights are controlled by the new parameter $\mu_w$. 
	As suggested in~\cite{2009LFRweighted},
	each node $i$ is given a target strength $s_i = k_i^{\beta}$ with weight exponent
	$\beta=1.5$, where $k_i$ is the degree of node $i$ drawn from the LFR power-law degree distribution. 
	A fraction $1-\mu_w$ of $s_i$ is assigned to edges within the same community,
	a fraction $\mu_w$ to external edges. The combination of parameters $\mu$ and $\mu_w$ makes it more or less difficult to recognize the true network communities. To distinguish the difficult from the easy instances, consider that, for each node, a fraction $1-\mu_w$ of its strength is assigned to its internal edges, which are a fraction $1-\mu$ of its overall edges. Therefore, the expected weight of an internal edge is proportional to $(1-\mu_w)/(1-\mu)$, while the expected weight of an external edge is proportional to $\mu_w/\mu$. The easiest instances are those in which weights are informative, that is, the instances in which $(1-\mu_w)/(1-\mu) > \mu_w/\mu$, a condition that is satisfied if and only if $\mu>\mu_w$. 
	
	We consider $\mu \in \{0.5, 0.6, 0.7, 0.8\}$, then apply 
	the weight mixing values $\mu_w \in \{0.1, 0.3, 0.5, 0.7\}$ to the same unweighted underlying graphs. We run the \milanoEx{}
	algorithm, controlling for $\alpha$ in $[0,1]$ in steps of $0.05$, with $10$ multistart restarts per instance.
	Recovery is measured by the AMI against the planted partition, averaged over the $20$ instances with the same parameter setting.
	
	\begin{table}[t]
		\centering
		\begin{tabular}{c c cccc}
			\toprule
			& unweighted & \multicolumn{4}{c}{weighted, by $\mu_w$}\\
			\cmidrule(lr){2-2}\cmidrule(lr){3-6}
			$\mu$ & (reference) & $0.1$ & $0.3$ & $0.5$ & $0.7$\\
			\midrule
			$0.5$ & $0.999$ \,(0.95) & $\mathbf{1.000}$ \,(1.00) & $\mathbf{1.000}$ \,(0.95) & $0.985$ \,(0.95) & $0.081$ \,(1.00) \\
			$0.6$ & $0.974$ \,(0.90) & $\mathbf{1.000}$ \,(1.00) & $\mathbf{1.000}$ \,(1.00) & $0.994$ \,(0.95) & $0.315$ \,(1.00) \\
			$0.7$ & $0.739$ \,(0.90) & $\mathbf{1.000}$ \,(1.00) & $0.999$ \,(1.00) & $0.986$ \,(0.95) & $0.526$ \,(1.00) \\
			$0.8$ & $0.172$ \,(0.85) & $\mathbf{0.971}$ \,(1.00) & $0.953$ \,(1.00) & $0.925$ \,(0.95) & $0.620$ \,(0.90) \\
			\bottomrule
		\end{tabular}
		\caption{
			Computational results on weighted graphs: optimal $\alpha^*$ for each $(\mu,\mu_w)$ configuration, with the corresponding AMI in parentheses.}
		\label{tab:weighted}
		
	\end{table}
	
	\begin{figure}[t]
		\centering
		\begin{subfigure}[b]{0.48\textwidth}
			\includegraphics[width=\textwidth]{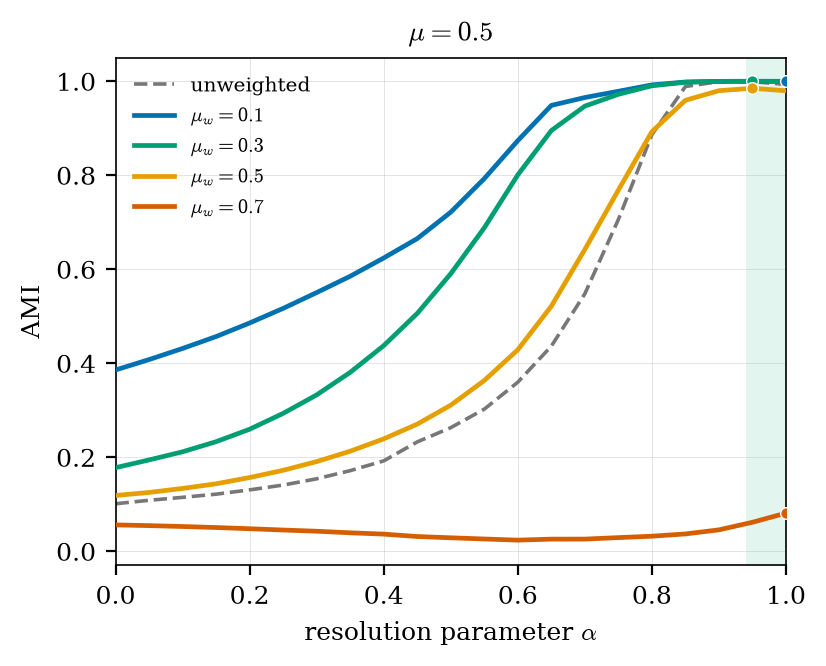}
			\caption{}
			\label{fig:weighted-curves-mu050}
		\end{subfigure}
		\hfill
		\begin{subfigure}[b]{0.48\textwidth}
			\includegraphics[width=\textwidth]{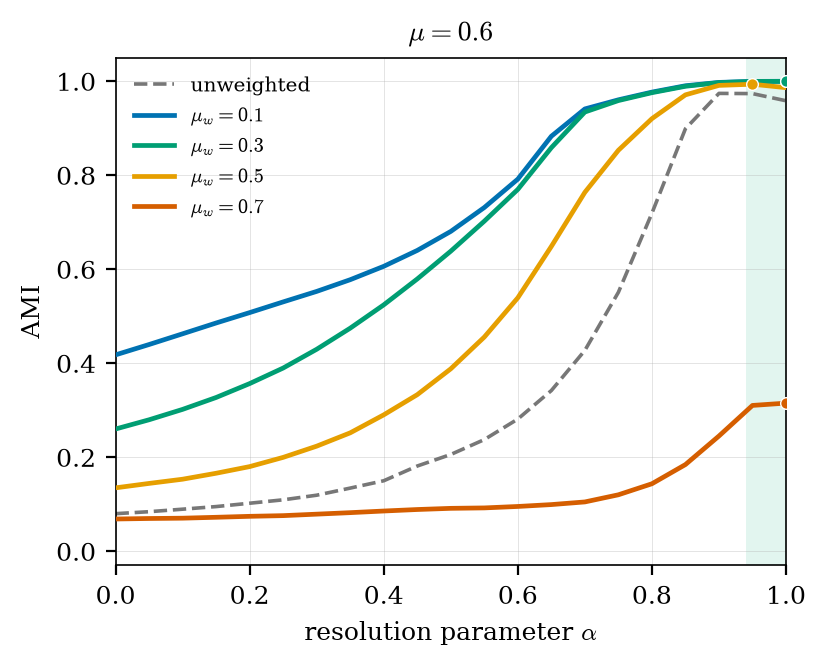}
			\caption{}
			\label{fig:weighted-curves-mu060}
		\end{subfigure}
		
		\vspace{0.6em}
		\begin{subfigure}[b]{0.48\textwidth}
			\includegraphics[width=\textwidth]{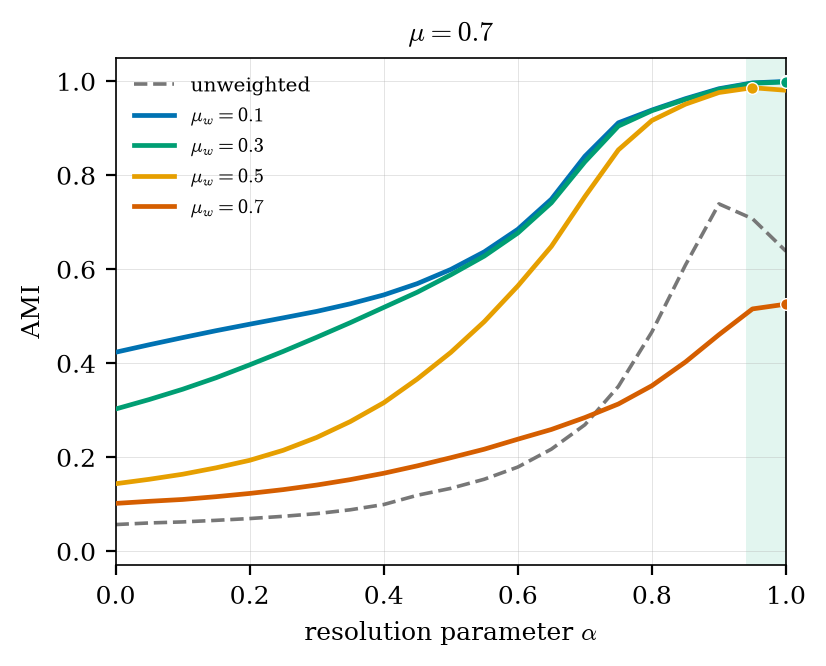}
			\caption{}
			\label{fig:weighted-curves-mu070}
		\end{subfigure}
		\hfill
		\begin{subfigure}[b]{0.48\textwidth}
			\includegraphics[width=\textwidth]{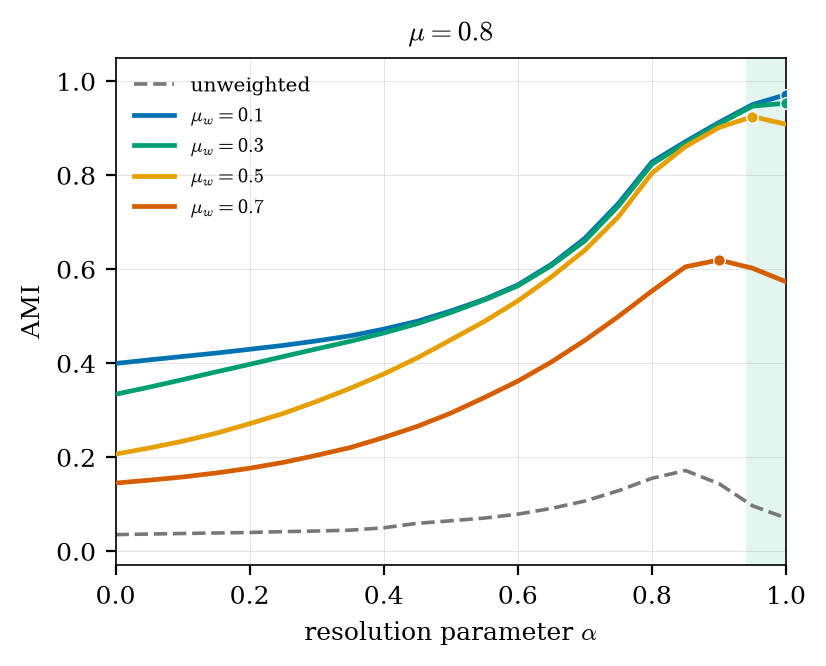}
			\caption{}
			\label{fig:weighted-curves-mu080}
		\end{subfigure}
		\caption{
			Stability of the optimal $\alpha^*$ across the $(\mu,\mu_w)$ grid. The shaded band is the optimal $\alpha^*$ range $[0.95,1.00]$.}
		\label{fig:weighted-curves}
	\end{figure}
	
	Table~\ref{tab:weighted} reports the AMI as a function of $\mu$ and $\mu_w$,
	together with the unweighted reference, while Figure~\ref{fig:weighted-curves} shows the corresponding AMI$(\alpha)$ curves, one panel for each value of $\mu$. When  $\mu > \mu_w$, weights generally improve the recovery of the planted
	communities relative to the unweighted reference. Consider, for example, the column corresponding to the weighted graph with $\mu_w=0.5$ and compare it to the column corresponding to the unweighted graph. There, the values of the AMI of the weighted graphs are better than those of the unweighted. Similar results are available in all cells of the table.
	Moreover, each panel of Figure~\ref{fig:weighted-curves} fixes one topological difficulty $\mu$ and overlays the four weight regimes on the unweighted reference: for $\mu_w<\mu$ the weighted curves lie above the reference and their peaks fall in the band $\alpha^*\in[0.95,1.00]$, slightly above the unweighted $\alpha^*$ ($0.85$--$0.95$), whereas the anti-informative case $\mu_w=0.7$ stays well below and rises above the reference only at $\mu=0.8$ (panel~d), the single panel where $\mu>\mu_w$.
	Overall, the weighted experiments confirm that the high-recovery region remains concentrated in a narrow range of $\alpha$.

	\subsection{Comparison with other community detection methods}
	\label{sec:comparison}
	
	In this section, we compare \SNAPEx{} with other community detection methods. Some of them, including \SNAPEx{}, let researchers tune a resolution parameter to obtain partitions of different cardinality. A fair comparison should control, as far as possible, for the number of detected communities. Therefore, for methods using a resolution parameter, we sweep the parameter in their feasible range, and retain partitions whose cardinalities are the closest to the true partition size $K$; parameter-free methods are reported at their natural cluster count.
	
	First, we compare \SNAPEx{} optimization against Louvain~\cite{blondel2008fast} and Leiden~\cite{traag2019louvain} under the modularity-with-resolution formulation (\texttt{RBConfiguration}) of Reichardt and Bornholdt~\cite{Reichardt2006}, and the Constant Potts Model~\cite{traag2011narrow}. All these models belong to the modularity-based family of quality-function optimization methods, but with a fundamental difference. Louvain and Leiden use the configuration model to compare the density of the internal edges, while CPM uses a fixed threshold. All models use a resolution parameter to obtain partitions of different size: so, in modularity maximization we swept $\gamma\in[0.2,2.0]$ in steps of $0.1$, in CPM we swept its density threshold over $[0.01,0.20]$ in steps of $0.01$, while in \SNAPEx{} we swept $\alpha\in[0,1]$ in steps of $0.05$; for every sweep, each method is optimized over $10$ restarts, and the reported AMI is averaged over the $20$ independent graphs generated for each configuration.
	
	Then we compare \SNAPEx{} against Markov Clustering (MCL), Infomap, Walktrap, and Label propagation (LPA). Among these methods, only Infomap is based on the optimization of an objective function (the Map equation), while the other three are constructive, in the sense that communities are generated through their respective algorithmic mechanisms rather than by optimizing an explicit global objective function. Moreover, only MCL provides a resolution parameter to output partitions of different sizes, i.e. the \textit{inflation parameter}, swept over $[1.1,2.5]$ in steps of $0.1$. The same setting is used for both unweighted and weighted networks, while the other methods are parameter-free and output a single partition.
	
	The comparison is performed for both unweighted and weighted benchmark networks.
	
	\paragraph{Unweighted networks} Tests are run on unweighted networks, simulated with the LFR procedure~\cite{2008LFR}. Parameters are fixed to $N=10\,000$, node degree exponent  $\tau_1=2.0$, average degree $\bar{k} = 25$ and maximum degree $80$. Next, we simulated three scenarios for the community-size distribution, varying the community-size exponent $\tau_2$ and the size range. The \textit{homogeneous} scenario is characterized by $\tau_2=2.5$ and size range  $[50,100]$, the \textit{intermediate} scenario by $\tau_2=2.0$ and size range $[30,300]$, the \textit{heterogeneous} scenario by $\tau_2=1.5$ and size range $[20,1000]$. For each scenario, the ratio $\mu$ of edges pointing outside a community is chosen in the set $\{0.4,0.5,0.55,0.65,0.75\}$. So we have $5\times3 = 15$ scenario/parameter configurations.
	Results about the AMI are reported in Table~\ref{tab:vs-multibaseline}. Considering first the four resolution-based methods in the homogeneous scenario, all four algorithms recover the planted partition up to $\mu=0.65$. A minor difference is that Leiden edges out Louvain, as expected since it is a more refined modularity heuristic.  
	As heterogeneity increases, differences among the four methods become more pronounced, and \SNAPEx{} generally achieves the highest AMI when methods can reliably recover the planted partition.
	There is a case, heterogeneous with $\mu = 0.75$, in which recovering the communities is the most difficult. Here we see that \SNAPEx{} is not the best; however the values of the AMIs are considerably low.
	
	A closer look at the table suggests that the advantage of \SNAPEx{} over the other methods tends to increase with the heterogeneity. For example, with $\mu = 0.65$, it outperforms its best competitor by $+0.005$ in the intermediate scenario, by $+0.081$ in the heterogeneous scenario. A possible explanation is that, although resolution parameters allow partitions of different sizes to be obtained, not all partition sizes are attainable under each objective function.

	In Figure~\ref{fig:vs-modularity}, we consider the heterogeneous scenario with $\mu = 0.65$. On the $x$-axis, we report the number of communities obtained for some value of the resolution parameter, and on the $y$-axis, the corresponding AMI. The vertical line with $K = 75$ corresponds to the true value of the partition size. 
	Leiden and Louvain produce at most $68$ and $67$ communities, respectively; that is, even though the resolution parameter is varied, neither method reaches the planted value $K = 75$. The reason for this could be that varying the parameter $\gamma$ is not enough to overcome the resolution limit. 
	A similar result is observed for CPM, the smallest attainable number of communities is $180$, well above the planted value $75$.  
	A possible explanation is that CPM relies on a constant density threshold, which may be less suited to distinguishing communities whose sizes span the broad range $[20,1000]$.
	Conversely, in this experiment the \SNAPEx{} sweep spans a much wider range of partition sizes than the competing objectives---from a few dozen to several thousand communities---and it is the only method that reaches the planted value $K=75$ (Figure~\ref{fig:vs-modularity}).
	
	\begin{figure}[t]
		\centering
		\includegraphics[width=0.72\textwidth]{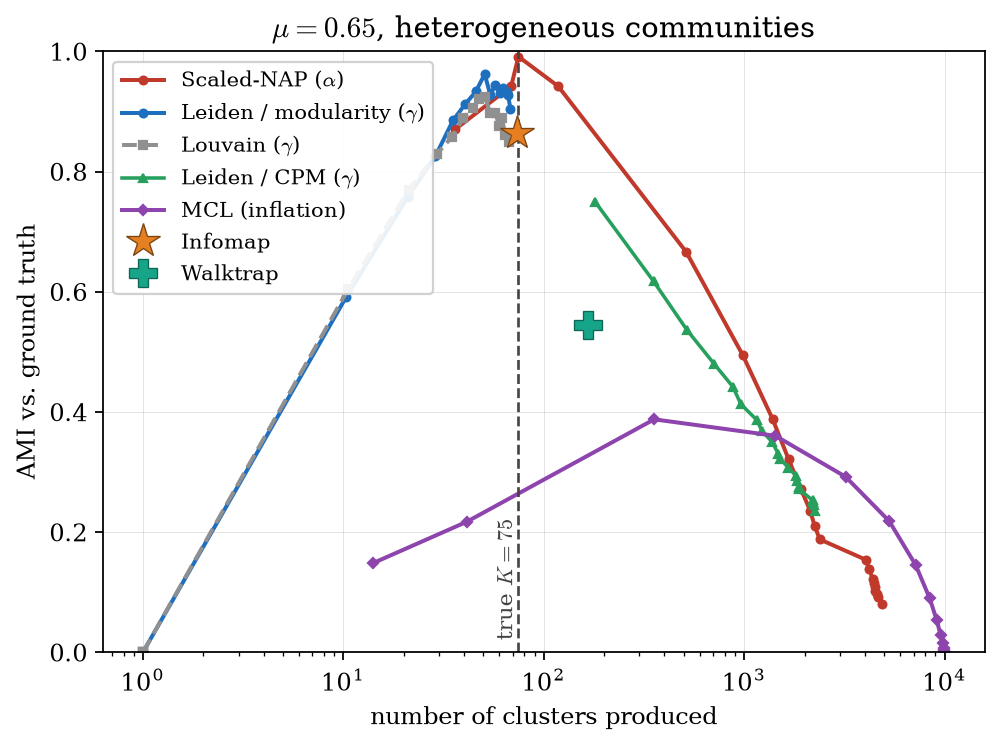}
		\caption{Representative cell: AMI as a function of the number of clusters produced, heterogeneous regime ($\mu=0.65$, community sizes in $[20,1000]$, true $K=75$). The vertical dashed line marks the true number of communities.}
		\label{fig:vs-modularity}
	\end{figure}

	Turning to the four additional methods, as can be seen in Table~\ref{tab:vs-multibaseline}, in the homogeneous scenario all methods perform well when community structure is easily detectable. As \(\mu\) increases, however, Infomap emerges as the strongest among the four additional methods: at $\mu=0.75$, it retains an AMI of $0.986$, while MCL and Walktrap drop substantially and LPA fails to recover the planted structure.
	The result is remarkable, as Infomap does not require a resolution parameter to be tuned. 
	The findings are confirmed in the intermediate scenario: Infomap remains highly competitive, but its performance drops sharply when recognition is the hardest, i.e. at $\mu = 0.75$. 
	The same happens in the heterogeneous scenario. 
	In Figure~\ref{fig:vs-modularity}, the number of communities and the corresponding AMI are reported. There, it can be seen that, for this specific instance, Infomap 
	produces a partition with $73$ communities, very close to the planted value $75$.
	Walktrap outputs a larger number of communities, while MCL, through its inflation parameter, can produce partitions with cardinalities both below and above $K$; however, its AMI remains below that of the best-performing methods. Finally, LPA is not reported as its AMI is zero. 
	
	Overall, the unweighted comparison shows that community structures at different levels of resolution can be recovered, and that, when the exact number of communities is known, \SNAPEx{} optimization performs well against a broad set of alternative community detection methods. However, it should be outlined that Infomap is a strong competitor and has the advantage of being parameter-free, although its performance deteriorates in some of the most difficult cases in which \SNAPEx{} still provides accurate recovery. 

	\begin{table}[t]
		\centering
		\setlength{\tabcolsep}{4pt}
		\begin{tabular}{|l l c |c c c c | c c c c|}
			\toprule
			heterogeneity & $\mu$ & $K$ & \SNAPEx{} & Leid. & Louv. & CPM & MCL & Info. & Walk. & LPA\\
			\midrule
			homogeneous   & 0.40 & 147 & \textbf{1.000} & \textbf{1.000} & \textbf{1.000} & \textbf{1.000} & \textbf{1.000} & \textbf{1.000} & \textbf{1.000} & \textbf{1.000}\\
			homogeneous   & 0.50 & 147 & \textbf{1.000} & \textbf{1.000} & \textbf{1.000} & \textbf{1.000} & \textbf{1.000} & \textbf{1.000} & 0.999 & 0.999\\
			homogeneous   & 0.55 & 147 & \textbf{1.000} & \textbf{1.000} & \textbf{1.000} & \textbf{1.000} & 0.999 & \textbf{1.000} & 0.993 & 0.999\\
			homogeneous   & 0.65 & 148 & \textbf{1.000} & \textbf{1.000} & 0.998 & \textbf{1.000} & 0.981 & \textbf{1.000} & 0.857 & 0.994\\
			homogeneous   & 0.75 & 147 & 0.976 & 0.989 & 0.953 & \textbf{0.992} & 0.452 & 0.986 & 0.422 & 0.000\\
			\midrule
			intermediate      & 0.40 & 131 & \textbf{1.000} & 0.998 & 0.998 & \textbf{1.000} & 0.953 & \textbf{1.000} & \textbf{1.000} & \textbf{1.000}\\
			intermediate      & 0.50 & 133 & \textbf{1.000} & 0.998 & 0.998 & \textbf{1.000} & 0.880 & \textbf{1.000} & 0.997 & \textbf{1.000}\\
			intermediate      & 0.55 & 129 & \textbf{1.000} & 0.997 & 0.997 & \textbf{1.000} & 0.828 & \textbf{1.000} & 0.980 & \textbf{1.000}\\
			intermediate      & 0.65 & 134 & 0.999 & 0.994 & 0.990 & 0.992 & 0.647 & \textbf{1.000} & 0.793 & 0.979\\
			intermediate      & 0.75 & 130 & \textbf{0.896} & 0.865 & 0.780 & 0.859 & 0.106 & 0.048 & 0.364 & 0.000\\
			\midrule
			heterogeneous & 0.40 &  76 & \textbf{1.000} & 0.997 & 0.997 & 0.951 & 0.651 & \textbf{1.000} & 0.996 & \textbf{1.000}\\
			heterogeneous & 0.50 &  68 & \textbf{1.000} & 0.997 & 0.997 & 0.893 & 0.442 & \textbf{1.000} & 0.933 & 0.911\\
			heterogeneous & 0.55 &  70 & \textbf{1.000} & 0.997 & 0.996 & 0.872 & 0.362 & \textbf{1.000} & 0.848 & 0.550\\
			heterogeneous & 0.65 &  75 & \textbf{0.991} & 0.910 & 0.848 & 0.750 & 0.217 & 0.864 & 0.544 & 0.000\\
			heterogeneous & 0.75 &  75 & 0.323 & \textbf{0.347} & 0.325 & 0.328 & 0.057 & 0.000 & 0.213 & 0.000\\
			\bottomrule
		\end{tabular}
		\caption{Computational results about the AMIs of different methods on unweighted graphs.} 
		\label{tab:vs-multibaseline}
		
	\end{table}
	
	\paragraph{Weighted networks} 
	The previous comparison between all community detection methods is then repeated for weighted graphs.
	The same scenarios, for the community size distribution, i.e. homogeneous, intermediate, and heterogeneous, are retained, while edge weights are generated using different values of $\mu_w$.  Preliminary results on the homogeneous and intermediate scenarios show only marginal differences among methods; therefore, we focus the discussion on the heterogeneous case, as the comparison is more informative.
	For this scenario, we consider
	$\mu\in\{0.4,0.5,0.55,0.65,0.75\}$ paired with $\mu_w\in\{0.1,0.3,0.5,0.7\}$ assigned on the same graph. 
	
	Results about the AMIs of each method are reported in
	Table~\ref{tab:weighted-vs}. The results are more nuanced than in the
	unweighted case, with no method uniformly dominating across all
	configurations. More informative differences emerge when the results are
	distinguished according to whether weights reinforce or contrast with the
	planted community structure.

	To interpret these results, we can compare the methods in three scenarios: the unweighted case, the case in which weights are informative about the real community structure, that is when $\mu_w < \mu$, and finally the case in which weights can be misleading, that is when $\mu_w > \mu$. When the graph is unweighted, on these network instances we see that the previous results are confirmed, as \SNAPEx{} is best in 4 cases out of 5. Leiden (modularity) is best when the mixing parameter $\mu$ is the highest, $\mu = 0.75$, but it is a case in which all AMIs are low. Infomap is good in the easiest instances, $\mu \le 0.55$, but very inaccurate in the hardest: AMI = 0.0 for $\mu = 0.75$.
	
	Looking at the cases in which weights are informative ($\mu_w < \mu)$, we can find that all methods do well. Even when a method is not best, such as \SNAPEx{} for $\mu = 0.4, \mu_w = 0.3$, the differences are negligible, with many methods recovering the true community structure. Finally, when weight may mislead the community structure ($\mu_w > \mu)$, we can see that Infomap suffers a sharp decline in its accuracy, while \SNAPEx{}, Leiden and Walktrap can still recognize a consistent partition.

	\begin{table}[t]
		\centering
		\small
		\setlength{\tabcolsep}{4.5pt}
		\begin{tabular}{|ccl|cccc|cccc|}
			\toprule
			$\mu$ & true $K$ & $\mu_w$ & \SNAPEx{} & Leiden & Louvain & CPM & MCL & Infomap & Walktrap & LPA \\
			\midrule
			\multirow{5}{*}{$0.4$} & \multirow{5}{*}{$72$} & unw. & \textbf{1.000} & 0.998 & 0.998 & 0.956 & 0.657 & \textbf{1.000} & 0.997 & \textbf{1.000} \\
			&  & $0.1$ & \textbf{1.000} & \textbf{1.000} & \textbf{1.000} & \textbf{1.000} & \textbf{1.000} & \textbf{1.000} & \textbf{1.000} & \textbf{1.000} \\
			&  & $0.3$ & 0.999 & 0.999 & 0.999 & \textbf{1.000} & 0.945 & \textbf{1.000} & \textbf{1.000} & \textbf{1.000} \\
			&  & $0.5$ & 0.993 & \textbf{0.995} & 0.992 & 0.697 & 0.457 & 0.882 & 0.931 & 0.362 \\
			&  & $0.7$ & 0.008 & 0.018 & 0.007 & 0.017 & 0.023 & 0.000 & \textbf{0.145} & 0.003 \\
			\midrule
			\multirow{5}{*}{$0.5$} & \multirow{5}{*}{$71$} & unw. & \textbf{1.000} & 0.997 & 0.997 & 0.899 & 0.439 & \textbf{1.000} & 0.936 & 0.961 \\
			&  & $0.1$ & \textbf{1.000} & \textbf{1.000} & \textbf{1.000} & \textbf{1.000} & \textbf{1.000} & \textbf{1.000} & \textbf{1.000} & \textbf{1.000} \\
			&  & $0.3$ & \textbf{1.000} & \textbf{1.000} & \textbf{1.000} & \textbf{1.000} & 0.945 & \textbf{1.000} & \textbf{1.000} & \textbf{1.000} \\
			&  & $0.5$ & 0.967 & 0.996 & 0.997 & 0.972 & 0.511 & \textbf{1.000} & 0.960 & 0.997 \\
			&  & $0.7$ & 0.038 & 0.074 & 0.037 & 0.063 & 0.037 & 0.000 & \textbf{0.257} & 0.000 \\
			\midrule
			\multirow{5}{*}{$0.55$} & \multirow{5}{*}{$69$} & unw. & \textbf{1.000} & 0.997 & 0.996 & 0.853 & 0.356 & \textbf{1.000} & 0.836 & 0.370 \\
			&  & $0.1$ & \textbf{1.000} & \textbf{1.000} & \textbf{1.000} & \textbf{1.000} & \textbf{1.000} & \textbf{1.000} & \textbf{1.000} & \textbf{1.000} \\
			&  & $0.3$ & \textbf{1.000} & \textbf{1.000} & \textbf{1.000} & \textbf{1.000} & 0.946 & \textbf{1.000} & \textbf{1.000} & \textbf{1.000} \\
			&  & $0.5$ & 0.973 & 0.997 & 0.997 & 0.975 & 0.542 & \textbf{1.000} & 0.957 & \textbf{1.000} \\
			&  & $0.7$ & 0.085 & 0.136 & 0.088 & 0.109 & 0.040 & 0.000 & \textbf{0.303} & 0.000 \\
			\midrule
			\multirow{5}{*}{$0.65$} & \multirow{5}{*}{$74$} & unw. & \textbf{0.991} & 0.887 & 0.835 & 0.753 & 0.213 & 0.893 & 0.547 & 0.000 \\
			&  & $0.1$ & \textbf{1.000} & \textbf{1.000} & \textbf{1.000} & \textbf{1.000} & \textbf{1.000} & 0.923 & \textbf{1.000} & \textbf{1.000} \\
			&  & $0.3$ & \textbf{1.000} & \textbf{1.000} & \textbf{1.000} & 0.999 & 0.956 & \textbf{1.000} & \textbf{1.000} & \textbf{1.000} \\
			&  & $0.5$ & 0.999 & 0.998 & 0.997 & 0.974 & 0.581 & \textbf{1.000} & 0.970 & \textbf{1.000} \\
			&  & $0.7$ & 0.342 & 0.443 & 0.347 & 0.269 & 0.063 & 0.000 & \textbf{0.464} & 0.000 \\
			\midrule
			\multirow{5}{*}{$0.75$} & \multirow{5}{*}{$70$} & unw. & 0.307 & \textbf{0.327} & 0.311 & 0.303 & 0.048 & 0.000 & 0.199 & 0.000 \\
			&  & $0.1$ & \textbf{1.000} & \textbf{1.000} & \textbf{1.000} & \textbf{1.000} & \textbf{1.000} & 0.396 & \textbf{1.000} & 0.891 \\
			&  & $0.3$ & \textbf{1.000} & \textbf{1.000} & \textbf{1.000} & 0.995 & 0.987 & 0.926 & 0.999 & 0.892 \\
			&  & $0.5$ & \textbf{1.000} & 0.999 & 0.999 & 0.969 & 0.648 & 0.942 & 0.975 & 0.905 \\
			&  & $0.7$ & 0.781 & \textbf{0.842} & 0.718 & 0.737 & 0.067 & 0.142 & 0.543 & 0.000 \\
			\bottomrule
		\end{tabular}
		\caption{Computational results about the AMIs of different methods on weighted graphs.}
		\label{tab:weighted-vs}
	\end{table}

	\subsection {Community detection on real-world networks} 
	\label{sec:real-world}
	
	LFR benchmarks provide networks with a known ground-truth community structure whose characteristics can be systematically controlled,
	but they do not reproduce the irregularities of real networks: strongly heterogeneous degree distributions, large numbers of small communities, partial overlaps, and incomplete ground-truth coverage.  
	To investigate whether the results obtained on LFR carry over to real data, we evaluate the \milanoEx{} algorithm on three \SNAPcoll networks with available ground truth~\cite{YangLeskovec2015}: \texttt{com-Amazon} (product co-purchasing), \texttt{com-DBLP} (scientific co-authorship), and \texttt{com-YouTube} (user friendships and groups). 
	The three networks range in size from $\sim\!3.3\times10^{5}$ to $\sim\!1.1\times10^{6}$ nodes,
	and therefore make it possible to assess \emph{quality} and \emph{scalability} jointly.
	
	The \SNAPcoll ground truth consists of \emph{overlapping} and \emph{partial} communities, as only a fraction of the nodes belong to a known community and some of them to more than one. Under these conditions,  standard measures such as NMI and AMI cannot be applied as they assume the community structure to be a complete, non-overlapping partition. Therefore, in the following experiments we will compare methods through the $F_1$ score,  
	see~\cite{YangLeskovec2013}. Given the predicted partition $\Pi=\{\mathcal C_1,\ldots,\mathcal C_q\}$ and the set of ground-truth communities $\mathcal{S}$, 
	and setting \[
	F_1(\mathcal C_i,S_j)
	=
	\frac{2|\mathcal C_i\cap S_j|}
	{|\mathcal C_i|+|S_j|},
	\]
	the overall $F_1$ score is defined as:
	\begin{equation}
		\label{Score}
		F_1
		=\frac{1}{2}\left(
		\frac{1}{|\mathcal S|}
		\sum_{S_j\in\mathcal S}
		\max_{\mathcal C_i\in\Pi}
		F_1(\mathcal C_i,S_j)
		+
		\frac{1}{|\Pi|}
		\sum_{\mathcal C_i\in\Pi}
		\max_{S_j\in\mathcal S}
		F_1(\mathcal C_i,S_j)
		\right)
	\end{equation}
	
	i.e.\ the average of the two best-matching directional scores (ground truth$\to$predicted and predicted$\to$ground truth). By construction, the $F_1$ is a value that ranges between $0$ and $1$, with the highest values corresponding to the greatest agreement between the detected and ground-truth communities. 
	
	We compare \milanoEx{} against seven other methods/algorithms. Each method is run in its native implementation: Louvain and Leiden with the modularity objective~\cite{blondel2008fast,traag2019louvain}, Leiden with the Constant Potts Model objective (\texttt{Leiden-CPM}), Infomap~\cite{rosvall2008maps}, Walktrap~\cite{pons2005computing}, Markov Clustering (MCL)~\cite{dongen2000graph}, and Label Propagation (LPA)~\cite{raghavan2007near}. For each method that exposes a tuning parameter, we perform a full sweep: 
	\milanoEx{} over $\alpha\in[0,1]$ with step $0.05$ (21 values); Leiden and Louvain over the resolution parameter $\gamma\in[0.2,2.0]$ with step $0.1$ (19 values); Leiden-CPM over its resolution parameter in $[0.01,0.20]$ with step $0.01$ (20 values); MCL over the inflation parameter in $[1.1,2.5]$ with step $0.1$ (15 values). For each method, the configuration achieving the highest $F_1$
	score over the parameter sweep is retained for comparison. Each sweep includes the default value of the corresponding native implementation ($\gamma=1.0$ for Louvain and Leiden, at which the objective coincides with plain Newman--Girvan modularity, and inflation $2.0$ for MCL). It is worth noting that the default implementation of Leiden-CPM has the parameter $\gamma=1.0$, that is not meaningful on sparse graphs. Indeed, in the CPM objective $\gamma$ acts as an internal edge-density threshold, so $\gamma=1.0$ would only accept near-cliques, and the useful range on networks of this density lies orders of magnitude lower --- as confirmed by the best-scoring values found ($0.05$ in the Amazon network, $0.19$ in DBLP). Walktrap, Infomap, and LPA do not expose a resolution parameter and are run once. Every randomized algorithm is run allowing for 10 restarts, the partition with the best internal quality is retained.
	
	Runs on the real-world networks share the hardware and the timeout/OOM conventions of Section~\ref{subsec:baselines}. The TO and OOM outcomes are stated explicitly and they are a sign of the scalability of the methods. Alongside wall-clock time, the driver records the CPU time of every execution (user plus system time of the child process, all threads included); the comparison of the two is reported below.
	
	Table~\ref{tab:real-graphs} summarizes the characteristics of the three networks  (\texttt{com-Amazon}, \texttt{com-DBLP}, and \texttt{com-YouTube}), where it can be seen that only a small fraction of the nodes are reported as belonging to at least one community. Table~\ref{tab:real-f1} reports the $F_1$ of all methods, together with the number of communities produced by the best configuration.

	\begin{table}[t]
		\centering
		\small
		\begin{tabular}{lrrrr}
			\toprule
			Network & $|V|$ & $|E|$ & K & covered nodes \\
			\midrule
			\texttt{com-Amazon} & 334\,863 & 925\,872 & 5\,000 & 16\,716 \\
			\texttt{com-DBLP} & 317\,080 & 1\,049\,866 & 5\,000 & 93\,432 \\
			\texttt{com-YouTube} & 1\,134\,890 & 2\,987\,624 & 5\,000 & 39\,841 \\
			\bottomrule
		\end{tabular}
		\caption{\SNAPcoll real-world networks: sizes and ground truth.}
		\label{tab:real-graphs}
		
	\end{table}
	
	\begin{table}[t]
		\centering
		\small
		\setlength{\tabcolsep}{5pt}
		\begin{tabular}{lccc}
			\toprule
			& \texttt{com-Amazon} & \texttt{com-DBLP} & \texttt{com-YouTube} \\
			Method & \small($3.3\!\times\!10^{5}$) & \small($3.2\!\times\!10^{5}$) & \small($1.1\!\times\!10^{6}$) \\
			\midrule
			\milanoEx{}
			& $0.485$ \tiny(23461)
			& $\mathbf{0.398}$ \tiny(56269)
			& $\mathbf{0.254}$ \tiny(260518) \\
			\midrule
			Walktrap
			& $\mathbf{0.506}$ \tiny(14905)
			& $0.363$ \tiny(30425)
			& OOM \\
			MCL
			& $0.484$ \tiny(21550)
			& $0.355$ \tiny(48297)
			& OOM \\
			LPA
			& $0.476$ \tiny(22427)
			& $0.384$ \tiny(23843)
			& $0.051$ \tiny(33398) \\
			Leiden-CPM
			& $0.468$ \tiny(37065)
			& $0.397$ \tiny(91187)
			& TO \\
			\midrule
			Leiden
			& $0.260$ \tiny(275)
			& $0.190$ \tiny(270)
			& $0.017$ \tiny(7575) \\
			Louvain
			& $0.200$ \tiny(116)
			& $0.158$ \tiny(637)
			& $0.018$ \tiny(5932) \\
			Infomap
			& $0.063$ \tiny(14)
			& $0.080$ \tiny(513)
			& $0.015$ \tiny(912) \\
			\bottomrule
		\end{tabular}
		\caption{Computational results on real networks: $F_1$ values, TO~$=$~timeout, OOM~$=$~out-of-memory.}
		\label{tab:real-f1}
		
	\end{table}
	
	\noindent Three observations emerge from Table~\ref{tab:real-f1}.
	First, on all three networks \milanoEx{} outperforms Leiden and Louvain by a wide margin: $0.485$ versus $0.260$/$0.200$ on Amazon, $0.398$ versus $0.190$/$0.158$ on DBLP, and $0.254$ versus $0.017$/$0.018$ on YouTube, with a gap ranging from a factor of $\sim\!2$ to more than one order of magnitude. The mechanism echoes what was observed on LFR: the ground truth is made up of many small communities (the top $5\,000$ per network), whereas modularity optimization, limited by its resolution, collapses to only a few hundred large communities ($116$--$637$) on the medium-sized networks --- and a few thousand on \texttt{com-YouTube} --- even when the whole range of the resolution parameter is explored. \milanoEx{}, by contrast, produces fine-grained partitions ($2.3\times10^{4}$--$2.6\times10^{5}$ communities) in which the small communities match the ground-truth more accurately.
	
	\noindent Second, although Infomap is competitive on LFR benchmarks (Section~\ref{sec:comparison}),
	its performance drops sharply on the three real-world networks considered here ($F_1$ of $0.063$, $0.080$, $0.015$), 
	producing very coarse partitions (as few as $14$ communities on Amazon).
	The strongest competitor on synthetic data is thus among the weakest on this real data.
	
	\noindent Third, on the two medium-sized networks (Amazon, DBLP), \milanoEx{} is among the best methods; it can be seen that both Leiden-CPM and Walktrap have a good ability to recognize the network structure (Walktrap is the best on Amazon, Leiden-CPM is second by a narrow margin on DBLP). However, the two methods cannot be used on the largest network, YouTube, in which memory or time limits are exceeded. More precisely: Walktrap and MCL exhaust the $128$~GB of RAM (OOM), Leiden-CPM exceeds the one-hour limit for a single execution (TO), and the quality of LPA drops from $0.476$/$0.384$ on the medium networks to $0.051$ (more on computational times in the next paragraph). As a conclusion, on real networks \milanoEx{} seems to be the best.  
	
	Finally, \milanoEx{} is also the fastest method of all: in Table~\ref{tab:real-times} the running times are reported. They are referred to a single configuration for each method (i.e.\ one $\alpha$ value for \milanoEx{}, one resolution value for Leiden/Louvain/Leiden-CPM, one inflation value for MCL, or the single configuration of Walktrap/Infomap/LPA), averaged over the parameter sweep where applicable. For every stochastic method the ten restarts of a configuration are included in the reported time.
	Times were measured in the Python driver script as wall-clock time around each execution (thus including input/output and process launch), so that all methods are timed by the same procedure, with \milanoEx{} run single-threaded (\texttt{-{}-threads=1}). The driver also records the CPU time of every execution (user plus system time of the child process, all threads included): for every method except Infomap the CPU-to-wall ratio is $1.00$ within $1\%$ on all three networks, confirming that the executions are effectively single-threaded and CPU-bound. The exception is Infomap, whose ratio of $2.0$, $1.8$, and $1.5$ on the three networks reveals internal multi-threading in its native implementation: its wall-clock times in Table~\ref{tab:real-times} are therefore favoured relative to the single-threaded protocol, which we note for fairness --- the ranking of Table~\ref{tab:real-f1} is unaffected.
	
	\begin{table}[t]
		\centering
		\caption{Computational times on real networks.}
		\label{tab:real-times}
		\small
		\setlength{\tabcolsep}{6pt}
		\begin{tabular}{lccc}
			\toprule
			& \texttt{com-Amazon} & \texttt{com-DBLP} & \texttt{com-YouTube} \\
			Method & \small($3.3\!\times\!10^{5}$) & \small($3.2\!\times\!10^{5}$) & \small($1.1\!\times\!10^{6}$) \\
			\midrule
			\milanoEx{} (persistence)
			& $\mathbf{15.1}$ & $\mathbf{15.7}$ & $\mathbf{76.2}$ \\
			\midrule
			Walktrap & $1234.4$ & $1835.4$ & OOM \\
			MCL & $181.0$ & $171.8$ & OOM \\
			LPA & $2119.1$ & $3744.2$ & $7139.5$ \\
			Leiden-CPM & $2097.3$ & $1871.2$ & TO \\
			\midrule
			Leiden & $1078.0$ & $799.4$ & $15343.8$ \\
			Louvain & $242.0$ & $348.9$ & $871.6$ \\
			Infomap$^{\dagger}$ & $113.2$ & $89.0$ & $435.7$ \\
			\bottomrule
		\end{tabular}
	\end{table}
	
	The computational times certify the scalability of the methods. On all three networks \milanoEx{} is the fastest method by a wide margin, despite running on a single thread: a single $\alpha$ configuration takes on average $15$--$16$~s on the medium-sized networks and only $76.2$~s on \texttt{com-YouTube}. By contrast, on \texttt{com-YouTube} a single Leiden configuration (ten restarts) takes on average $1.5\times10^{4}$~s, i.e.\ roughly $200\times$ slower than \milanoEx{}, and even the single LPA configuration costs $7\,140$~s, about $94\times$ slower --- for a solution with an $F_1$ score about five times lower. 
	The timings also explain the failures reported in Table~\ref{tab:real-f1}: Leiden-CPM already requires $\sim\!1\,900$--$2\,100$~s per configuration on the medium-sized networks, and on \texttt{com-YouTube} a single execution exceeds the one-hour limit (TO), while Walktrap and MCL exhaust memory (OOM). The modest 
	growth of the persistence running time with network size ($\sim\!5\times$ from Amazon to YouTube, against a $3.2\times$ growth in edges) is, in practice, what makes \milanoEx{} usable at the scale where the fine-grained baselines become inapplicable.

	\section{Conclusions}
	\label{sec:conclusions}
	
	We introduced the \emph{Scaled Null-Adjusted Persistence} (\SNAPEx{}), a parametrized measure to assess the quality of a network partition, that includes the null-adjusted persistence and the modularity at the two extremes of its parametrization. The two measures are obtained by controlling the cluster volumes through the parameter $\alpha\in[0,1]$: NAP is obtained at $\alpha=0$ and modularity at $\alpha=1$. Using $\alpha$ as a resolution knob, a researcher can move from small and fine-grained communities, as is typical of NAP, to large and coarse communities as typical of modularity.
	
	We derived analytical properties of the family---including conditions under which two clusters are merged---and we showed how it can be applied to real networks. We developed \milanoEx{}, a scalable Louvain-style heuristic for its optimization; on the LFR benchmarks the accuracy-maximizing exponent $\alpha^\star$ proved remarkably stable, lying in the narrow band $[0.85,0.95]$ across all configurations, and \SNAPEx{} attained the highest AMI in every cell in which the planted structure is recoverable. On real-world networks with up to $1.1\times10^{6}$ nodes it was best or tied for best in best-matching $F_1$ on \texttt{com-DBLP} and \texttt{com-YouTube} and competitive on \texttt{com-Amazon}, while running fastest by a wide margin on a single thread.
	
	A further direction of research is the design of more sophisticated algorithms for \SNAPEx{} optimization. As experienced for modularity maximization, various interchange and long-term search heuristics can be implemented, improving the objective function and the partition quality at the cost of greater computational time. The \SNAPEx{} equation can also be elaborated further, to analyse directed networks or the cases in which communities can overlap. Finally, some domain-specific networks are worth analysing through \SNAPEx{}, for example financial systems, a case in which the persistence of the information flow within a group is economically meaningful.
	
	\bibliographystyle{plain}
	\bibliography{references}

	\appendix
	
	\section{An exact mixed-integer linear programming formulation}
	\label{app:ilp}
	
	In this appendix we show that, for a
	fixed value of $\alpha$, the \SNAPEx{} maximization
	problem~\eqref{snap_optimization} can be formulated as a mixed-integer linear program (MILP).
	This is of practical interest on small instances, where a general-purpose MILP
	solver returns a \emph{provably optimal} partition and thus an exact reference
	against which the \milanoEx heuristic of Section~\ref{sec:milano} can be validated. The
	formulation does not remove the intrinsic difficulty of the problem: since
	$\sNAP{1}(\Pi)=Q(\Pi)$ by~\eqref{nap_modularity_relation}, maximizing \SNAPEx{}
	at $\alpha=1$ is modularity maximization, which is NP-hard~\cite{brandes2008}, so
	the model is expected to be solvable only for networks of intermediate size.
	
	Throughout this section $\alpha\in[0,1]$ is fixed and we keep the notation
	introduced in the previous sections, together with the compact form of the
	objective in~\eqref{snap_weighted_compact},
	\begin{equation*}
		\sNAP{\alpha}(\mathcal C)
		=
		\frac{2\,w_{\mathrm{in}}(\mathcal C)\,\operatorname{vol}(\mathcal C)^{\alpha-1}}{(2S)^{\alpha}}
		-
		\frac{\operatorname{vol}(\mathcal C)^{\alpha+1}}{(2S)^{\alpha+1}} .
	\end{equation*}
	The only structural assumption we use is that the edge weights are integer valued
	(rational weights can be rescaled to integers), so that
	every cluster volume $\operatorname{vol}(\mathcal C)=\sum_{i\in\mathcal C}s_i$
	takes values in the finite set $\{0,1,\dots,V_{\max}\}$, with $V_{\max}=2S$. This is
	what lets us replace the non-integer powers $\operatorname{vol}(\mathcal
	C)^{\alpha\pm1}$ by \emph{precomputed constants}. 
	
	
	\paragraph{Decision variables.}
	We allow up to $K$ communities indexed by $c\in[K]=\{1,\dots,K\}$; taking $K=n$ is
	without loss of generality. For $v\in\{1,\dots,V_{\max}\}$ precompute the constants
	\begin{equation}
		\label{eq:ilp-coeffs}
		\gamma^{\mathrm{in}}_v=\frac{2\,v^{\,\alpha-1}}{(2S)^{\alpha}},
		\qquad
		\gamma^{\mathrm{vol}}_v=\frac{v^{\,\alpha+1}}{(2S)^{\alpha+1}},
		\qquad
		\gamma^{\mathrm{in}}_0=\gamma^{\mathrm{vol}}_0=0 .
	\end{equation}
	The variables are $x_{ic}\in\{0,1\}$ ($i\in V$, $c\in[K]$), equal to $1$ if node
	$i$ is assigned to community $c$; $y_{ijc}\in\{0,1\}$ ($\{i,j\}\in E$, $c\in[K]$),
	equal to $1$ if both endpoints of $\{i,j\}$ lie in $c$; $z_{cv}\in\{0,1\}$
	($c\in[K]$, $v\in\{0,\dots,V_{\max}\}$), equal to $1$ if community $c$ has volume
	exactly $v$; and $t_{cv}\ge 0$, which linearizes the product $W_c\,z_{cv}$
	(see~\eqref{eq:ilp-mccormick} below). We use the linear expressions
	\begin{equation}
		\label{eq:ilp-aux}
		W_c=\sum_{\{i,j\}\in E} w_{ij}\,y_{ijc}\ \bigl(=w_{\mathrm{in}}(\mathcal C_c)\bigr),
		\qquad
		\mathrm{Vol}_c=\sum_{i\in V} s_i\,x_{ic}\ \bigl(=\operatorname{vol}(\mathcal C_c)\bigr),
		\qquad U=S\ \ (W_c\le U).
	\end{equation}
	
	\paragraph{The model $(\mathrm{P}^{\mathrm{MILP}}_\alpha)$.}
	\begin{align}
		\max_{x,y,z,t}\quad
		& \sum_{c\in[K]}\Bigl(\ \sum_{v=1}^{V_{\max}}\gamma^{\mathrm{in}}_v\,t_{cv}
		\;-\;\sum_{v=1}^{V_{\max}}\gamma^{\mathrm{vol}}_v\,z_{cv}\ \Bigr)
		\label{eq:ilp-obj}\\[2pt]
		\text{s.t.}\quad
		& \sum_{c\in[K]} x_{ic}=1, & \forall i\in V, \label{eq:ilp-assign}\\
		& y_{ijc}\le x_{ic},\quad y_{ijc}\le x_{jc},\quad y_{ijc}\ge x_{ic}+x_{jc}-1,
		& \forall\{i,j\}\in E,\ \forall c, \label{eq:ilp-and}\\
		& \sum_{v=0}^{V_{\max}} z_{cv}=1, & \forall c, \label{eq:ilp-onehot}\\
		& \sum_{v=0}^{V_{\max}} v\,z_{cv}=\mathrm{Vol}_c, & \forall c, \label{eq:ilp-volsel}\\
		& t_{cv}\le U z_{cv},\ \ t_{cv}\le W_c,\ \ t_{cv}\ge W_c-U(1-z_{cv}),\ \ t_{cv}\ge 0,
		& \forall c,\ \forall v. \label{eq:ilp-mccormick}
	\end{align}
	Constraint~\eqref{eq:ilp-assign} makes the assignment a partition;
	\eqref{eq:ilp-and} is the standard linearization forcing
	$y_{ijc}=x_{ic}\wedge x_{jc}$, so that $W_c=w_{\mathrm{in}}(\mathcal C_c)$;
	\eqref{eq:ilp-onehot}--\eqref{eq:ilp-volsel} select the (integer) volume of each
	community through the one-hot vector $z_{c\cdot}$; and the McCormick
	inequalities~\eqref{eq:ilp-mccormick} enforce $t_{cv}=W_c\,z_{cv}$, i.e.\
	$t_{cv}=W_c$ for the selected volume and $t_{cv}=0$ otherwise. The nonlinear terms
	of the objective are thereby recovered exactly: if community $c$ has volume
	$v^\star$, then $\sum_v\gamma^{\mathrm{in}}_v t_{cv}
	=2W_c(v^\star)^{\alpha-1}/(2S)^{\alpha}
	=2w_{\mathrm{in}}(\mathcal C_c)\operatorname{vol}(\mathcal C_c)^{\alpha-1}/(2S)^{\alpha}$
	and $\sum_v\gamma^{\mathrm{vol}}_v z_{cv}
	=\operatorname{vol}(\mathcal C_c)^{\alpha+1}/(2S)^{\alpha+1}$.
	
	\begin{proposition}[Exactness]
		\label{prop:ilp-exact}
		For any fixed $\alpha\in[0,1]$ and any $K$ at least as large as the number of
		communities in an optimal partition of~\eqref{snap_optimization}, the optimal
		value of $(\mathrm{P}^{\mathrm{MILP}}_\alpha)$ equals
		$\max_{\Pi}\sNAP{\alpha}(\Pi)$, and every optimal $\{x_{ic}\}$ induces an
		optimal partition through $\mathcal C_c=\{i:x_{ic}=1\}$. Taking $K=n$ is always
		sufficient.
	\end{proposition}
	\begin{proof}
		Any feasible integer point yields, through~\eqref{eq:ilp-assign}, a partition of
		$V$ into the nonempty classes among $\mathcal C_1,\dots,\mathcal C_K$.
		By~\eqref{eq:ilp-and}, $W_c=w_{\mathrm{in}}(\mathcal C_c)$; by~\eqref{eq:ilp-volsel}
		the unique $v$ with $z_{cv}=1$ equals $\operatorname{vol}(\mathcal C_c)$; and
		\eqref{eq:ilp-mccormick} gives $t_{cv}=W_c$ for that $v$ and $0$ otherwise.
		Substituting into~\eqref{eq:ilp-obj} and using~\eqref{eq:ilp-coeffs} reproduces
		$\sum_c \sNAP{\alpha}(\mathcal C_c)=\sNAP{\alpha}(\Pi)$ term by term (empty
		communities select $v=0$ and contribute $0$). Conversely, any partition with at
		most $K$ classes is representable, so the two optima coincide.
	\end{proof}
	
	\begin{remark}[Redundant inequalities and size]
		\label{rem:ilp-tighten}
		Since $\gamma^{\mathrm{in}}_v\ge 0$ and the model maximizes, $W_c$ (hence each
		$y_{ijc}$ and $t_{cv}$) is pushed upward at optimality: the lower bound
		$y_{ijc}\ge x_{ic}+x_{jc}-1$ in~\eqref{eq:ilp-and} and the lower McCormick bound
		$t_{cv}\ge W_c-U(1-z_{cv})$ in~\eqref{eq:ilp-mccormick} are then redundant and may
		be dropped to speed up the solver; standard symmetry-breaking constraints
		(e.g.\ ordering the communities by smallest node index) help as well. The
		formulation uses $O(nK)$ assignment, $O(mK)$ edge and $O(K\,V_{\max})$
		volume/linearization variables, i.e.\ it is \emph{pseudo-polynomial} in the total
		volume $2S$; it is therefore practical only on small, integer-weighted instances,
		which is precisely the regime of exact benchmarks.
	\end{remark}
	
	\section{Comparison between AMI, NMI and ARI}
	\label{app:measures}
	
	The comparison between the community detection methods has been carried out using the AMI index. In this appendix we motivate the choice by proving that:
	\begin{itemize}
		\item the AMI corrects the NMI, in the sense that the latter is biased towards positive values even in the cases in which no community structure is recognized;
		\item the AMI is strongly correlated with the Adjusted Rand Index (ARI)~\cite{Hubert1985}, as their results are quite similar, therefore only results about the former are reported, to avoid information redundancy.
	\end{itemize}
	
	All results in this appendix are computed on the simulated LFR networks of Section~\ref{sec:benchmark}, calculated through the three scenarios, homogeneous, $\tau_2=2.5$, sizes $[50,100]$, intermediate, $\tau_2=2.0$, $[30,300]$, heterogeneous, $\tau_2=1.5$, $[20,1000]$ (all with $N=10\,000$, $\tau_1=2.0$, $\bar{k}=25$, maximum degree $80$).
	
	The commonly used Normalized Mutual Information (NMI) is not corrected for chance, as is the AMI, and therefore systematically biased to positive values even in cases in which no community is recognized. For example, in Table~\ref{tab:measure-degeneracy} we report data about the MCL algorithm in the homogeneous scenario with $\mu=0.5$. As reported, the MCL algorithm depends on a parameter to obtain communities at different resolution level. The simulated data are characterized by $K=147$ true communities, that are recovered for a particular value of that parameter. As can be seen, the three measures ARI, NMI, and AMI are all equal to 1, proving that the communities correspond to the ground truth. 
	When the parameter is tuned, more communities are output and their distances to the ground truth are increasing, but at a different pace. Finally, the last partition is quite inconsistent, as it groups $10\,000$ nodes into $9\,030$ clusters, most of them being singletons. However, the NMI value has remained quite high, $0.710$, suggesting that this partition does have some information value even though it is not consistent. Conversely, the AMI and the ARI correctly score it at values $0.081$ and $0.013$, values that are close to 0.

	In the three scenarios (homogeneous, intermediate, heterogeneous) the method rankings using the ARI are the same as those obtained using the AMI, the only difference being that intervals are wider using the ARI. 
	In Table~\ref{tab:ari-heterogeneous} only data about the heterogeneous scenario are reported (as they are the most discriminating). There, it can be seen that the rankings through the ARI are the same as those discussed for Table~\ref{tab:vs-multibaseline}, but the difference is more marked. For example, the margin of \SNAPEx{} against CPM grows from $+0.049$ to $+0.118$ at $\mu=0.4$, from $+0.127$ to $+0.264$ at $\mu=0.55$, and from $+0.242$ to $+0.446$ at $\mu=0.65$, almost double in every case. The reason is structural from the definition of the index: the ARI counts agreements over node \emph{pairs}, so a partition that merges or splits large communities is penalised far more heavily than by a node-based index such as the AMI. Precisely because the errors of the modularity family (merging small communities) and of CPM (fragmenting large ones) are of this kind, the ARI registers them more sharply.
	
	\begin{table}[t]
		\centering
		\begin{tabular}{c r ccc}
			\toprule
			inflation & \# communities & AMI & ARI & NMI \\
			\midrule
			$1.5$ & $147$ & $1.000$ & $1.000$ & $1.000$ \\
			$1.7$ & $782$ & $0.895$ & $0.751$ & $0.928$ \\
			$1.8$ & $1984$ & $0.766$ & $0.534$ & $0.864$ \\
			$1.9$ & $3559$ & $0.615$ & $0.338$ & $0.811$ \\
			$2.0$ & $5204$ & $0.458$ & $0.194$ & $0.771$ \\
			$2.1$ & $6632$ & $0.318$ & $0.102$ & $0.744$ \\
			$2.4$ & $9030$ & $0.081$ & $0.013$ & $0.710$ \\
			\bottomrule
		\end{tabular}
		\caption{Comparison between AMI, ARI and NMI on MCL algorithm outcomes at various resolution levels ($K = 147$ is the true number of communities).}
		\label{tab:measure-degeneracy}
		
	\end{table}

	\begin{table}[H]
		\centering
		\setlength{\tabcolsep}{4pt}
		\begin{tabular}{l l c c c c c c c c c}
			\toprule
			heterogeneity & $\mu$ & $K$ & \SNAPEx{} & Leid. & Louv. & CPM & MCL & Info. & Walk. & LPA\\
			\midrule
			heterogeneous & 0.40 &  76 & \textbf{1.000} & 0.999 & 0.999 & 0.882 & 0.404 & \textbf{1.000} & 0.993 & \textbf{1.000}\\
			heterogeneous & 0.50 &  68 & 0.999 & 0.999 & 0.999 & 0.772 & 0.175 & \textbf{1.000} & 0.862 & 0.709\\
			heterogeneous & 0.55 &  70 & \textbf{1.000} & 0.999 & 0.995 & 0.735 & 0.115 & \textbf{1.000} & 0.666 & 0.333\\
			heterogeneous & 0.65 &  75 & \textbf{0.986} & 0.847 & 0.732 & 0.540 & 0.037 & 0.653 & 0.210 & 0.000\\
			heterogeneous & 0.75 &  75 & 0.073 & \textbf{0.105} & 0.092 & 0.101 & 0.003 & 0.000 & 0.043 & 0.000\\
			\bottomrule
		\end{tabular}
		\caption{Computational results about the ARIs of different methods on unweighted graphs.}
		\label{tab:ari-heterogeneous}
	\end{table}
	
\end{document}